\documentclass[12pt,reqno,twoside]{amsart}
\usepackage{lineno}
\usepackage{amsmath}
\usepackage{amsthm} 
\usepackage{mathtools}
\usepackage{bm}
\usepackage{graphicx}
\usepackage{esvect}
\usepackage{fancyhdr}
\usepackage{graphicx}
\usepackage{graphics}
\usepackage{amssymb}
\usepackage{enumitem}
\usepackage{geometry}
\usepackage{braket}
\usepackage{xcolor}
\usepackage[colorlinks = true]{hyperref}
\usepackage{subfigure}
\usepackage{tikz}
\usetikzlibrary{decorations.pathreplacing} 
\usepackage{multirow}
\usepackage{tabularx}
\usepackage{enumerate}
\usepackage{tcolorbox}
\definecolor{egg}{rgb}{.98,.97,.92}
\definecolor{astroorange}{rgb}{1,.93,.79}
\definecolor{darkorange}{rgb}{1,.89,.6}
\definecolor{dullblue}{rgb}{.29,.47,.77}
\definecolor{grayblue}{rgb}{.98,.98,.98}
\definecolor{fadedblue}{rgb}{.78,.86,.92}
\definecolor{tiffanyblue}{rgb}{.96,1,1}
\definecolor{grayish}{rgb}{.93,.93,.97}
\definecolor{charcoal}{rgb}{.247,.259,.27}
\definecolor{evergreen}{rgb}{.7725,.858,.7647}
\definecolor{dullred}{rgb}{.929,.498,.598}
\definecolor{lavender}{rgb}{.8,.741,.85}

\newcommand{\qhofa}{q\nobreakdash-HOFA}

\newcommand{\mbe}{\mathbb{E}}

\newcommand{\Tr}[1]{\mathrm{Tr}\left\{#1\right\}}
\newcommand{\Ptr}[2]{\mathrm{Tr}_{#1}\left\{#2\right\}}

\newcommand{\ot}{\mathrm{OTOC}}

\newcommand{\norm}[1]{\left|\left| #1 \right|\right|}
\newcommand{\abs}[1]{\left| #1 \right|}

\newcommand{\lra}[1]{\left\langle #1 \right\rangle}

\newcommand{\lrp}[1]{\left(  #1 \right)}

\renewcommand{\geq}{\geqslant}
\renewcommand{\leq}{\leqslant}
\renewcommand{\ge}{\geqslant}
\renewcommand{\le}{\leqslant}

\def\ot{\otimes}
\def\complex{\mathbb{C}}

\def\Z{\mathbb{Z}}

\newcommand{\proj}[1]{| #1\rangle\!\langle #1 |}

\newcommand{\inner}[2]{\langle #1 , #2\rangle}

\newcommand{\be}{\begin{equation}}
\newcommand{\ee}{\end{equation}}
\newcommand{\beq}{\begin{eqnarray}}
\newcommand{\eeq}{\end{eqnarray}}

\newtheorem{thm}{Theorem}
\newtheorem{lem}[thm]{Lemma}
\newtheorem{prop}[thm]{Proposition}
\newtheorem{cor}[thm]{Corollary}
\newtheorem{Def}[thm]{Definition}
\newtheorem{Con}[thm]{Conjecture}

\newtheorem{Exam}[thm]{Example}
\newtheorem{Rem}[thm]{Remark}

\title{Quantum Higher-Order Fourier Analysis\\ and the Clifford Hierarchy}

\author{Kaifeng Bu$^{1,2}$}

 \address{$^1$\textnormal{Department of Mathematics, The Ohio State University, Columbus, Ohio 43210, USA}}
 \email{\href{mailto:bu.115@osu.edu}{bu.115@osu.edu} (Kaifeng Bu)}

\author{Weichen Gu$^{1}$}
\email{\href{mailto:gu.1213@osu.edu}{gu.1213@osu.edu} (Weichen Gu)}

\author{Arthur Jaffe$^{2}$}
\address{$^2$\textnormal{Harvard University, Cambridge, MA 02138, USA}}
\email{\href{mailto:Arthur\_Jaffe@harvard.edu}{Arthur\_Jaffe@harvard.edu} (Arthur Jaffe)}

\begin{document}
\begin{abstract}
We propose a mathematical framework that we call \textit{quantum, higher-order Fourier analysis}. This generalizes the classical theory of higher-order Fourier analysis, which led to many recent advances in number theory and combinatorics. We define a family of ``quantum measures'' on linear transformations  on a Hilbert space, that reduce in the case of diagonal matrices  to the uniformity norms introduced by Timothy Gowers.  We show that our quantum measures and our related theory of quantum higher-order Fourier analysis characterize the Clifford hierarchy, an important notion of complexity in quantum computation. 
 In particular, we give a necessary and sufficient analytic condition that a unitary is an element of the $k^{\rm th}$ level of the Clifford hierarchy.
\end{abstract}

\maketitle



\bigskip

\section{Introduction}
This paper revolves about two interrelated themes.  First we present the foundation for a mathematical technique that we call quantum, higher-order Fourier analysis (\qhofa). 
We introduce a family of quantum uniformity measures $\norm{\ \cdot \ }_{Q^{k}}$ on linear transformations $B$,  given explicitly in \eqref{eq:QUn}. We show that $\norm{B}_{Q^{k}}$ is non-negative and increases monotonically in $k$. We prove that these measures are norms for $k=2,3$. These properties of the measures are sufficient to establish the resulting applications to quantum information that we consider here; the norm property has not proved necessary.

These measures generalize the classical uniformity norms introduced by Gowers~\cite{Gowers98},  which led to the theory of classical, higher-order Fourier analysis. The classical theory has been studied widely and provides many useful insights into classical combinatorics and number theory. The present paper provides an  analysis of the corresponding quantum theory. 

Our second theme  is to show that {\qhofa} gives an analytical measure of the complexity of quantum gates. A widely used algebraic classification of the complexity of a unitary gate in quantum information theory is the Clifford hierarchy, introduced by Gottesman and Chuang~\cite{GCh99}. This hierarchy has been extensively studied. We show here that our quantum uniformity measures quantify the Clifford hierarchy analytically. Among other things, we prove that a unitary $U$ belongs to the $k^{\rm th}$-level of the Clifford hierarchy, if and only if $\norm{U}_{Q^{k+1}}=1$. 

We believe that the analytic notions of {\qhofa} will find  other applications  in mathematical physics, as well as in mathematics and in quantum theory.

\subsection{Background}
In classical Fourier analysis one pairs the function $f(x)$  with the linear 
exponential $e^{ipx}$. One obtains higher-order, classical Fourier analysis by generalizing this procedure to pair  $f(x)$ with the exponential $e^{iP_k(x)}$ of a polynomial $P_{k}(x)$ of degree $k$~\cite{Gowers98, Gowers01,tao2012higher, tao2006additive}. Here $P_k(x)$ is called a \textit{phase polynomial}.  

One can characterize the degree $k$ of $P_{k}$ inductively, by using a discrete, non-linear, multiplicative derivative $\partial_{a}$. Suppose that a function $f(x)=e^{ig(x)}$ that takes values on the unit circle in the complex plane. In classical analysis, one defines the discrete multiplicative derivative as $\partial_{a}f(x) :=f(x+a)\overline{f(x)}=e^{i\lrp{g(x+a)-g(a)}}$. In case that $g(x)=P_{k+1}(x)$ is a polynomial of degree $k+1$, its derivative is a polynomial of degree $k$. 

Gowers quantified pairing with a phase polynomical by introducing  a family of uniformity norms~\cite{Gowers98}, that provide a remarkable, quantitative bound for Szemerédi’s theorem on arithmetic progressions~\cite{Gowers01}.  
These classical norms $\norm{\ \cdot \ }_{U^{k}}$ have the form 
\begin{align}\label{eq:UG}
\norm{f}^{2^{k}}_{U^{k}}
=\mathbb{E}_{\vec a_1,..., \vec a_{k}, \vec x}
(\partial_{\vec a_1,.., \vec a_k}f(\vec x))\;,
\end{align}
where the expectation $\mathbb{E}$ is a uniform average over the $\vec a_j$ and $\vec x$, where $\vec a_j, \vec x\in \mathbb{F}^n_2$ and $f$ is function defined on $\mathbb{F}^n_2$, namely $n$ copies of the finite field with two elements.

The further development of that work  led to  other applications in  number theory, in additive combinatorics, and in theoretical computer science. Many mathematicians including Gowers, Green, Tao, Szegedy, Host, B.~Kra and Ziegler, made significant contributions to build a comprehensive  mathematical theory of higher-order Fourier analysis; see,  e.g.~\cite{Gowers01,
host2005nonconventional,ziegler2007universal,tao2010inverse,tao2010inverse_low, GTZAnnal12, Host2017higher,Taoannal23,green-tao08}. The classical Gowers 3-norm has been shown to study vector stabilizer states and non-stabilizer states~\cite{Bu19,labib2022stabilizer,arun2024poly,mehr2024}.

Here we extend the ideas in classical, higher-order Fourier analysis to the quantum case. We study the classical configuration space of quantum theory as a non-commutative Weyl group. We define quantum uniformity measures on transformations on this space in \eqref{eq:QUn} and study the properties of these measures.
Discrete classical analysis corresponds to restricting the Weyl group to a certain commutative subgroup, in which case our quantum measures reduce to Gowers' classical uniformity norms. After studying properties of the quantum uniformity measures, we show their relevance to complexity theory in quantum computation.
\subsection{Quantum Derivative}
In this paper we replace the underlying space $\mathbb{F}_{2}$ by a non-commutative (quantum) phase space, equipped with a corresponding family of non-commutative translations. In continuum quantum physics, the Weyl group describes such a phase space, and it has been widely studied. In the next section we introduce a discrete quantum version of this framework that is the basis for much work in quantum information theory. In a different context, Gross~\cite{Gross06} (and  others) have extensively studied this discrete group as a quantum phase space.

We use unitary translations on the quantum phase space to define a discrete quantum derivative. Using this derivative, we define polynomials of degree $k$ inductively. This leads one to introduce quantum higher-order Fourier analysis, which involves the pairing of a unitary $U$  with the exponential of a polynomial of degree $k$.  This set of ideas provides the path to quantum, higher-order Fourier analysis.  

This mathematical framework works across different settings where translations define a discrete directional  derivative. Define the translation of a unitary $B$ in the direction given by a unitary $A$ as the conjugation $ABA^{*}$. One then obtains the discrete, multiplicative, directional derivative $ \partial_{A}B$ of $B$ as 
\begin{equation}\label{DiscreteDerivative}
  \partial_{A}B=(ABA^{*} )B^{*} =: [A,B]_{c} \;.
\end{equation}
Throughout this paper we denote the commutator $[\ ,\ ]_{c}$ by the expression \eqref{DiscreteDerivative}, which agrees with the group commutator in case $A,B$ are both unitaries.

\subsection{Weyl Operators and Clifford Unitaries}
Unitary matrices $X$ and $Z$ play a central role in quantum theory. One denotes the eigenbasis (or computational basis) for $Z$ by $\ket{k}$ where $k\in\Z_d$, and defines $Z\ket{k}=\omega^k\ket{k}$. Here  $\omega=e^{\frac{2\pi i}{d}}$. Also
$X\ket{k}=\ket{k+1}$, interpreted modulo $d$.
These transformations do not commute, but satisfy the relations,
$$[X^{q},Z^{p}]_{c}
= X^{q} Z^{p} X^{-q} Z^{-p}
= \omega^{-qp}\;. 
$$
For $ a=(p,q)\in V=\mathbb{Z}_{d}^{2}$, and for $\zeta=e^{\frac{\pi i (d+1)}{d}}$ a square root of $\omega$, define the unitary Weyl operator $w(a)$ as
\begin{align}\label{OneQuditWeyl}
w(a) = \zeta^{-pq} Z^{p} X^{q},
\end{align}
These unitaries generate the 1-qudit Weyl group, known as the Pauli group in the quantum literature.  Taking an $n$-fold tensor product yields the $n$-qudit Weyl group.

The Weyl operators are an orthonormal Fourier basis  with respect to the inner product 
\begin{align}\label{InnerProduct-n}  
\lra{A,B}=\frac{1}{d^n}\Tr{A^{*}B}\;,
\quad\text{so}\quad
A = \sum_{\vec a\in \Z^{2n}_{d}} \Xi_{A}(\vec a) \, w(\vec a)\; . 
\end{align}
Conjugation by the unitary $A=w(\vec a)$ defines  translation on  quantum  phase space by  $\vec a$, 
\begin{align}\label{QuantumTranslation}
B\mapsto B(\vec a) = w(\vec a)\,B \, w(\vec a)^{*}\; .
\end{align}

\subsection{Stabilizers and Classical Computation}
There is an important family of quantum states introduced by Gottesman~\cite{Gottesman97} and  defined  as common eigenstates of a commuting subgroup of Weyl operators. They are called ``stabilizer states;'' see Definition \ref{Def:Stabilizers}. The elegant mathematical structure of stabilizer states allows them to be precisely described and manipulated. This makes them invaluable in the development of quantum error correction, where they form the basis of stabilizer codes, a class of codes capable of protecting quantum information from errors. Notable examples of stabilizer codes include Shor’s 9-qubit code~\cite{ShorPRA95} and Kitaev’s toric code~\cite{Kitaev_toric}.  Along with the stabilizer states, there is an important family of unitaries called Clifford unitaries, which map Weyl operators by conjugation to Weyl operators. 
 
 The famous Gottesman-Knill theorem is the statement that quantum circuits with stabilizer states as input, Clifford unitaries as evolution, and measurement in the computational basis, can be efficiently simulated by a classical computer~\cite{gottesman1998heisenberg}.  This result delineates a boundary between classical and quantum computational power, and emphasizes the importance of non-Clifford gates in achieving quantum computational advantage.  Later, many other results on classical simulation have been considered in the literature~\cite{BravyiPRL16,BravyiPRX16,bravyi2019simulation,BeverlandQST20,SeddonPRXQ21, bu2022classical,Bu19,koh2015further,FLLW25}. Stabilizer states are the quantum equivalent of classical Gaussians~\cite{BGJ23a,BGJ23b,lyu2024fermionicgaussiantestingnongaussian}.
 
 One way to quantify the complexity of non-Clifford gates is called the Clifford hierarchy, introduced by Gottesman and Chuang~\cite{GCh99}.
This classification helps to implement fault-tolerant quantum computation~\cite{GCh99}. 
In this work, we explore these ideas through the lens of quantum higher-order Fourier analysis.

\subsection{Quantum Uniformity Measures}
Using the translation implemented by the Weyl operators, we define the phase-space derivative,  
\begin{align}\label{QuantumMultiplicativeDerivative}
\fbox{$\partial_{\vec a}B :=\partial_{w(\vec a)}B
= B(\vec a)B^{*}$}\;,
\end{align}
as  specified in  \eqref{DiscreteDerivative}, \eqref{QuantumTranslation}.  Multiple derivatives are defined iteratively as $\partial_{\vec a_k, \vec a_{k-1}, \ldots,\vec a_1} = \partial_{\vec a_k}\partial_{\vec a_{k-1},  \ldots,\vec a_1}$. 
The quantum uniformity measures are bounded by the operator norm: 
\begin{eqnarray}\label{eq:QUn}
\fbox{$\norm{B}^{2^{k}}_{Q^k}
=\mathbb{E}_{\vec a_k,...,\vec a_{1}\in V^n}\,
\frac{1}{d^n}
\Tr{ \partial_{\vec a_k,...,\vec a_{1}}B} \leqslant \norm{B}^{2^{k}}$}\;.
\end{eqnarray} 
The expectation $\mathbb{E}$ denotes an average over each $\vec a_{j}\in\Z_{d}^{2n}$. We develop properties of these measures in  \S\ref{Sec:QuantumUniformityNorm}. In  \S\ref{Sect:Convolution} we show that these measures can also be characterized by the quantum convolution proposed earlier by the authors in~\cite{BGJ23a,BGJ23b,BJ2025a,BGJ2025a}.  

We compare some properties of the classical and quantum Fourier analysis in Table~\ref{sum:tab1}. This theory gives rise to quantum phase polynomials of degree-$k$, generalizing  classical phase polynomials that occur in classical combinatorics. The quantum uniformity norms and the Clifford hierarchy, on which we elaborate here,  provide a natural starting point for a more general study of q-HOFA. 

\begingroup
\setlength{\tabcolsep}{6pt} 
\renewcommand{\arraystretch}{1.5} 

\begin{table}[htbp]
\centering
 \resizebox{.48\textwidth}{20mm}{
\centering
\begin{tabular}{  |c|c|c| }
\hline
& quantum, higher-order Fourier analysis  &  classical, higher-order Fourier analysis \\
\hline
{Derivative} & $\partial_AB=(ABA^*)B^* $ & $\partial_{\vec w}f(\vec x)=f(\vec x+\vec w)\overline{f(\vec x)}$ \\
\hline
Degree-1 polynomial & Weyl operator $w(\vec a)$  & phase linear function $(-1)^{\vec a\cdot \vec x}$\\
\hline
Fourier coefficient &  $\Xi_U[\vec a]=\inner{U}{w(\vec a)}$ & $\hat{f}(\vec a)=\inner{f}{(-1)^{\vec a\cdot x}}$ \\
\hline
Degree-k polynomial & $k^{\rm th}$  level Clifford hierarchy $\mathcal{C}^{(k)}$ & phase polynomial $(-1)^{g(\vec x)}$, $deg(g)\leq k$, denote the set as $D_k$\\
\hline 
Higher-order Fourier coefficient & $\inner{U}{V}, \forall V\in C^{(k)}$ &
$\inner{f}{g}$, $\forall g\in D_k$
\\
\hline
Uniformity measures & quantum uniformity measure $\norm{\cdot}_{Q^k}$ (Definition~\ref{Def:QGow}) & Gowers norm $\norm{\cdot}_{U^k}$ \eqref{eq:UG} \\
\hline
Monotonicity & $\norm{\cdot}_{Q^k}\leq \norm{\cdot}_{Q^{k+1}}$ (Theorem~\ref{Thm:Monotonicity}) &$\norm{\cdot}_{U^k}\leq \norm{\cdot}_{U^{k+1}}$ \\
\hline
Connection with $L^{p}$ norm & $\norm{\cdot}_{Q^k}\leq \norm{\cdot}_{p_k}$ with $p_k=\frac{2^k}{k+1}$ (Theorem~\ref{thm:rel_lp}) &$\norm{\cdot}_{U^k}\leq \norm{\cdot}_{p_k}$ with $p_k=\frac{2^k}{k+1}$\\
\hline
Log-convexity & $\norm{\rho}_{Q^2}\leq \norm{\rho}^{1/2}_{Q^1}\norm{\rho}^{1/2}_{Q^3}$, for pure state $\rho$ (Theorem \ref{thm:log_conv})& $\norm{f}_{U^2}\leq \norm{f}^{1/2}_{U^1}\norm{f}^{1/2}_{U^3}$ for positive function $f$\\
\hline
Application in property testing &  Clifford-hierarchy testing (Theorem~\ref{thm:QU_CH}) & Low-degree polynomial testing\\
\hline
\end{tabular}}
\vskip 8pt
\caption{A comparison of q-HOFA with c-HOFA}
\label{sum:tab1}
\end{table}
\endgroup
\goodbreak

\subsection{Summary of Main Results}
\begin{thm}[\bf Positivity,  Monotonocity, and Norm]\label{Thm:Pos-Mono-Norm}
For any $n$-qudit linear operator $B$,  
\begin{eqnarray*}
0\leqslant\norm{B}_{Q^k}\leq 
\norm{B}_{Q^{k+1}},\quad \text{for} \ 1\leqslant k\;.
\end{eqnarray*}
Furthermore, $\norm{\ \cdot \ }_{Q^k}$ is a norm for $k=2,3$.
\end{thm}

In Proposition \ref{Prop:Positivity}, Proposition \ref{Prop:InductiveDefn}, and Section \S\ref{Sect:MoreProperties}
we prove this and other relations for the quantum uniformity measures, including the two following propositions:

\begin{prop}\label{Prop:NormReductions}
For arbitrary $B$,
\begin{align*}\label{inductive}
\norm{B}_{Q^1}& = \abs{\frac{1}{d^n} \Tr{B} }\;,\quad
\norm{B}_{Q^2} = \norm{\Xi_{B}}_{{\ell^{4}}} \;,\\
\norm{B}^{2^{k}}_{Q^{k}}
&=\mathbb{E}_{\vec a}
\norm{\partial_{\vec a}B}^{2^{k-1}}_{Q^{k-1}}
=\underset{\vec a_1,...,\vec a_{k-1}\in V^n}{\mathbb{E}}\left|\frac{1}{d^n}\Tr{\partial_{\vec a_{k-1},...,\vec a_{1}}B}\right|^2  \;, k\geq 2\;.
\end{align*}
\end{prop}

\begin{cor}
If $\partial_{\vec {a}_{k-1}, \ldots, \vec {a}_{1}} B=\lambda^{2^{k-1}} I$ is a scalar multiple of the identity, then $\norm{B}_{Q^{k'}}=\lambda$, for $k'\geq k$. 
\end{cor}

\begin{prop} The $Q^{k}$ measure  of a pure state $\rho=\sum_{\vec a\in V^n}  \Xi_{\rho}(\vec a)\,  w(\vec  a)$ is given by $\ell^{p}$ norms of its Fourier coefficients $\Xi_{\rho}(\vec a)$. See Proposition \ref{Prop:Q34Pure} for the exact statement.
\end{prop}

The Clifford hierarchy $\mathcal{C}^{(k)}$, see Definition \ref{Def:CLi_Hier}, is an important algebraic concept used
to characterize the complexity of unitaries in quantum computation. Our quantum uniformity measures give an analytic characterization of the Clifford hierarchy.

\begin{thm}[\bf Analytic Characterization of ${\mathcal{C}^{(k)}}$] 
\label{Thm:CliffordHierarchy}
A unitary  $U\in\mathcal{C}^{(k)}$,  iff\, $\norm{U}_{Q^{k+1}}=1$.
\end{thm}
We complete the proof of Theorem \ref{Thm:CliffordHierarchy} in Proposition \ref{Prop:AnalyticCliffordCharacterization}.  While $\norm{U}_{Q^{1}}=\norm{U^{*}}_{Q^{1}}$,  in general for  for $k>1$, the unitaries $U$ and $U^{*}$ have different quantum uniformity norms $\norm{\ \cdot \ }_{Q^{k}}$. In Corollary \ref{Cor:AdjointOK} we give a class of unitaries $U\in\mathcal{C}^{(k)}$ which are closed under taking their adjoint.  

\begin{thm}[\bf The Classical Case] \label{Thm:TheClassicalCase}
Let $U$ be diagonal in the computational basis, with diagonal values $f$. Then  
$\norm{U}_{Q^k}=\norm{f}_{U^k}$,  the  Gowers' uniformity norm. 
\end{thm}

Next we address a different aspect of quantum Fourier analysis, namely the magnitude of the overlap between a unitary and $\mathcal{C}^{(k)}$.

\begin{Def}[\bf Overlap with the Hierarchy]
\label{Thm:HierarchyOverlap}
 The hierarchy-overlap measure of a unitary $U$ with the $k$-th level of the Clifford hierarchy
is 
\begin{eqnarray*}
\norm{U}_{q^{k+1}}
=\max_{V\in C^{(k)}}
|\inner{V}{U}|^2.
\end{eqnarray*}
\end{Def}

We have the following relation between that hierarchy-overlap measure and the quantum uniformity norms:

\begin{thm}[\bf Direct Inequality]
\label{Thm:DirectBound}
Given an $n$-qudit unitary $U$,  
\begin{eqnarray*}
\norm{U}_{q^k}
\leq \norm{U}_{Q^k}, \quad \text{for } 1\leqslant k\;.
\end{eqnarray*}
Furthermore $\norm{U}_{Q^k}=1$, iff $\norm{U}_{q^k}=1$. Thus  $U\in\mathcal{C}^{(k)}$, iff $\norm{U}_{q^{k+1}}=1$.
\end{thm}

\begin{Con}[\bf Inverse Inequality Conjecture]
If $0<c\leqslant\norm{U}_{Q^{k}}$, 
does there exist $d$, independent of $n$, such that  $0<d\leqslant\norm{U}_{q^{k}}$?
\end{Con}

\section{Preliminaries}
An $n$-qudit system is based on a Hilbert space $\mathcal{H}^{\ot n}$, where $\mathcal{H}\simeq \complex^d$. The chosen natural number $d$ characterizes the system and we take $d$ to be prime. 
Let $L(\mathcal{H}^{\ot n})$ denote the set of all linear transformations  on $\mathcal{H}^{\ot n}$,
and let 
$D(\mathcal{H}^{\ot n})$ denote the set of  all quantum states on $\mathcal{H}^{\ot n}$.

Let $a=(p,q)\in V$
denote a point in the phase space $V$.
For $d$ an odd prime, the one-qudit Weyl operators \eqref{OneQuditWeyl} is 
\begin{align}
w(a)=\omega^{-2^{-1}pq}Z^pX^q\;,
\end{align}
where $2^{-1}=\frac{d+1}{2}$ denotes the inverse of $2$ in $\mathbb{Z}_d$.
If $d=2$, take the Weyl operators to be 
\begin{eqnarray}
w(a)=i^{-pq}Z^pX^q \;.
\end{eqnarray}
The Weyl operators satisfy 
$w(a)^{*} = w(-a)$,
and their multiplication is given by  a symplectic product $\langle a,a' \rangle_{s}$ on phase space, $\lra{a,a'}_{s}
= pq'-qp'$.
Then 
\begin{align}\label{WeylMultiplication}
w(a)\,w(b)w(a)^{*}
&= \omega^{\inner{a}{b}_s}\,
 w(b)\;,
\end{align}
\begin{align} \label{Weyl2-M}
T(a):=w(a)Tw(a)^{*}
&=\sum_{b} \omega^{\inner{a}{b}_s}\,
\Xi_{T}(b) w(b)\;, \\
(T(a))(b)&=T(a+b)\;.
\end{align}
In the case of $n$ qudits, denote 
$
V^n=\mathbb{Z}^n_d\times\mathbb{Z}^n_d \; 
$.  For $\vec a\in V^n$. Define the Weyl operators  $w(\vec{a})$ as tensor products,
\begin{eqnarray*}
w(\vec a)
=w(a_1)\ot...\ot w(a_n) \;.
\end{eqnarray*}
As stated in the introduction, 
\begin{prop}[\bf Weyl Orthogonality]\label{Prop:WeylOrthonormal}
The Weyl operators $\set{w(\vec{a})}_{\vec a\in V^n}$ form an orthonormal basis 
in $L(\mathcal{H}^{\ot n})$ with respect to the inner product \eqref{InnerProduct-n}.
\end{prop}

The generators $w(\vec a)$ of the Weyl group are  our quantum  Fourier basis.  
Any linear operator $T$
on $\mathcal{H}^{\ot n}$ has a Fourier  expansion 
\be\label{T-Fourier}
T = \sum_{\vec a} \Xi_{T}\lrp{\vec a} w(\vec a)\;,
\quad\text{where \ }
\Xi_{T}\lrp{\vec{a}} = \lra{w(\vec a),T}\;.
\ee
The coefficients $\Xi_{T}(\vec a)$ are called the characteristic function of the  $n$-qudit transformation $T$ on $\mathcal{H}^{\ot n}$.  
This Fourier expansion has extensive
uses in a myriad of applications, including quantum boolean functions ~\cite{montanaro2010quantum},  classical simulation of quantum circuits~\cite{Bu19}, 
the quantum circuit complexity~\cite{Bucomplexity22}, nonlocal games based on noisy entangled states~\cite{Yao2019doubly, Yao21}, the sample complexity of quantum machine learning~\cite{BuPRA19_stat,bu2023effects} and many other ways. (See also a more general framework of quantum Fourier analysis in~\cite{JaffePNAS20}. )

\begin{Def}[\bf Quantum Translation and Expectation] \label{def:charfn}
The quantum translation of a tranformation $T$ is   $T(\vec a)=w(\vec a)\,T\,w(\vec a)^{*}$. The expectation $\mathbb{E}_{\vec a}$ of the $d^{2n}$ possible translations is the unweighted average, 
\be\label{ExpectationDef}
\mathbb{E}_{\vec a}\,T(\vec a)
= \frac{1}{d^{2n}}\sum_{\vec a\in V^{n}} T(\vec a)\;.
\ee
\end{Def}

\begin{prop} The expectation \eqref{ExpectationDef} over translations equals the normalized trace,
\be\label{ErgodicIdentity}
\mathbb{E}_{\vec a}\,T(\vec a)
= \frac{1}{d^n}\Tr{T}\;.
\ee
\end{prop}

\begin{proof}
For $\vec a =(\vec p, \vec q)$ and $\vec {b}=(\vec {p}',\vec{q}')$, the multiplication law \eqref{WeylMultiplication} gives
\begin{align*}
\mathbb{E}_{\vec a}\,w(\vec a) \,w(\vec b)\,w(\vec a)^{*}
= \frac{1}{d^{2n}} \sum_{\vec a}
\omega^{\lra{\vec{a},\vec{b}}_s}\,w(\vec{b})
 = \lrp{\prod_{j=1}^n 
\frac{1}{d^2} \sum_{p_j,q_j}
\omega^{\lrp{p_jq'_j-q_jp'_j}} 
}\,w(\vec{b})\nonumber
=  \delta_{\vec b,\vec 0}\;.
\end{align*}
Thus $
\mathbb{E}_{\vec a} T(\vec a)
=  \sum_{\vec{b}\in V}  \Xi_{T}(\vec b)\,
\delta_{\vec b,\vec 0} 
= \frac{1}{d^{n}}\Tr{T}$.
\end{proof}

\begin{Def}[\bf Stabilizer state~\cite{Gottesman96, Gottesman97}]\label{Def:Stabilizers}
A stabilizer vector for an $n$-qudit system is a vector that is invariant under a maximal abelian subgroup of Weyl operators. The  corresponding pure stabilizer state is the projection onto this eigenvector. A general stabilizer state $\rho$ is a  convex linear combination of pure stabilizer states.
\end{Def}

Stabilizer states encompass a broad class of quantum states, including Bell states and GHZ states, which play a fundamental role in quantum entanglement theory and quantum error correction theory. Several alternative characterizations of stabilizer states are  given in \cite{BGJ23a,BGJ23b}.

\begin{prop}[\cite{BGJ23a,BGJ23b}]\label{prop:stab_char}
    Given an $n$-qduit pure state $\rho=\proj{\psi}$, 
    it is a pure stabilizer state if and only if the absolute value of 
    the characteristic  function $|\Xi_{\rho}[\vec a]|\in \left\{0,\frac{1}{d^n}\right\}$, for any $\vec a\in V^n$.
\end{prop}

In addition to the stabilizer states, there is another important family of unitaries in quantum information called Clifford unitaries.

\begin{Def}[\bf Clifford unitary]
An  $n$-qudit unitary $U$ is a Clifford unitary if conjugation by $U$ maps every Weyl operator to another Weyl operator, up to a phase.
\end{Def}
It is elementary to see that conjugation by a Clifford unitary maps a stabilizer state to a stabilizer state.

To further understand the computational power of quantum gates beyond Clifford unitaries, the concept of the Clifford hierarchy was introduced by Gottesman and Chuang~\cite{GCh99}. 
The hierarchy consists of an increasing family of transformations $C^{(k)}\subset C^{(k+1)}$, where $C^{(1)}$ denotes the group generated by the $n$-qudit Weyl operators $w(\vec a)$.  

\begin{Def}[\bf Clifford hierarchy]
\label{Def:CLi_Hier}
For $k>1$, the $k^{\rm th}$ level $\mathcal{C}^{(k)}$ of  the Clifford hierarchy for $n$ qudits  is the set of unitaries on $\mathcal{H}^{\otimes n}$ such that 
\be
C^{(k)}
:=\left\{U:\, UC^{(1)} U^{*} \subset C^{(k-1)} \right\}\;.
\ee
\end{Def}
In case that $k=2$, the set $C^{(2)}$ is known as the set of all $n$-qudit  Clifford unitaries. 
In case $k=3$, the third level of the hierarchy $C^{(3)}$ includes the $T$ gate and the control-control $CCZ$ gate. These unitaries are used to realize universal quantum computation. Levels of the hierarchy reflect increasing complexity of a quantum gate $U$.

Aside from unitary transformations, 
a quantum  operation (or channel) is defined as a completely-positive and trace-preserving (CPTP) map. 

\begin{Def}[\bf Choi-Jamiołkowski isomorphism \cite{Choi75,Jamio72}]
Given a   quantum
channel $\Lambda$ from $\mathcal{H}_{A}$ to $\mathcal{H}_{A'}$,  the Choi state $J_{\Lambda}$ is
\begin{eqnarray*}
J_{\Lambda}=id_{A}\ot \Lambda (\proj{\Phi}) \;,
\end{eqnarray*}
where $| \Phi \rangle = \frac{1}{\sqrt{d^n}} \sum_{\vec j \in \mathbb{Z}_d^n}\ket{\vec j}_{A}\ot\ket{\vec j}_{A'} $ is the maximally entangled state on $\mathcal{H}_A\ot\mathcal{H}_{A'}$.
\end{Def}
For any input state $\rho$, the output state of the quantum channel $\Lambda(\rho)$ can be represented via the Choi state $J_{\Lambda}$ as 
\begin{eqnarray}\label{0127shi2}
\Lambda(\rho)
=d^n\Ptr{A}{J_{\Lambda} \cdot \rho^T_{A}\ot I_{A'}} \;,
\end{eqnarray}
where $\rho^T_{A}$ denotes the transpose of $\rho_A$.

\section{The Quantum Derivative}
Given two operators $A$ and $B$, one can consider the multiplicative, quantum derivative of $B$ in the direction $A$ is \eqref{DiscreteDerivative}, as  
$\partial_AB
    =ABA^{*}B^{*}
$.
Throughout this paper we specialize to the case that $A=w(\vec a)$ is a Weyl operator, whose conjugation implements a translation in quantum phase space, which we denote as $B(\vec a)=w(\vec a)B w(\vec a)^{*}$. This is a generalization of the notion of the classical discrete derivative, defined by translations on a classical configuration space.

\begin{Def}[\bf Quantum Discrete Derivative]{\label{Def:QuantumDerivative}} The multiplicative quantum derivative of $B$ in the direction $\vec a$ is
\be
\partial_{\vec a}B
=[w(\vec a), B]=w(\vec a) B w(\vec a)^{*}B^{*}
= B(\vec a)B^{*}\;.
\ee 
The corresponding $k$-th order derivative is given  inductively as
\be
\partial_{\vec a_k, ..., \vec a_1}B
=\partial_{\vec a_k} (\partial_{\vec {a}_{k-1}, ... ,{\vec {a}_1}}   B).
\ee
\end{Def}
While the multiplicative quantum derivative $\partial_{\vec a}B$ is not linear in $B$, nor does it satisfy the Leibniz rule.  However, the normalized trace satisfies  
\be\label{TraceDerivativeProduct}
    \frac{1}{d^n}
    \Tr{\partial_{\vec a}(B^*C)}
    =\lra{\partial_{\vec a}B, \partial_{\vec a}C}\;.
\ee

\begin{prop}[\bf Expectations of Products and Derivatives] \label{Prop:ExpectationDerivative}
The expectation defines a non-negative form, and the expectation of the quantum derivative is non-negative. In particular,
\be
\mathbb{E}_{\vec a} \lra{C, B(\vec a)}
= \frac{1}{d^n}\Tr{B} \  \overline{\frac{1}{d^n}\Tr{C}}\;,
\ee
and
\be\label{Q1-Positive}
0\leq\mathbb{E}_{\vec a} \, \partial_{\vec a}B=\abs{\frac{1}{d^{n}} \Tr{B}}^2\;.
\ee
\end{prop}

\begin{proof}
The first identity is a consequence of the definition \eqref{InnerProduct-n}  and \eqref{ErgodicIdentity}. The second equality follows from setting $B=C$ and using Definition \ref{Def:QuantumDerivative} of the quantum derivative.
\end{proof}

\paragraph{Reduction to the Classical Case:}
The quantum derivative reduces to the classical discrete derivative for operators that are diagonal in the computational basis $\ket{\vec x}$.  Let $f:\mathbb{Z}^n_d\to \complex$ be a function defining the diagonal operator
\be\label{DiagonalOperator}
B_f=\sum_{\vec x \in \mathbb{Z}_{d}^{n}  }f(\vec x)\proj{\vec x}\;.
\ee
The multiplicative, classical derivative in direction $\vec q$ of a complex-valued function  $f(\vec x)$ on $\mathbb{Z}_d^n$ is 
\be\label{ClassicalDerivative}
(\partial_{\vec q}f)(\vec x)=f(\vec x+\vec q)\overline{f(\vec x)}\;.
\ee
If $f$ takes values on the unit circle, this is a natural definition of a discrete derivative in direction $\vec q$ of the phase of $f$.  For $f$ taking arbitrary complex values, this definition is still useful, and one also  uses the name ``derivative.''

\begin{prop}
\label{prop:red_der}
Let $B_f$ have the form  \eqref{DiagonalOperator}, and let $\vec a=(\vec p, \vec q)$. The quantum derivative $\partial_{\vec a}B_f$  is 
\begin{eqnarray}
\partial_{\vec a}B_f
=B_{\partial_{-\vec q}f}\;.
\end{eqnarray}
\end{prop}

\begin{proof} 
By definition, the phase in $\omega(\vec a)$ cancels and 
\begin{align*}
    \partial_{\vec a}B_f
=&
Z^{\vec p}X^{\vec q}B_f X^{-\vec q}Z^{-\vec p}(B_f)^{*} \\
=& \sum_{\vec x,\vec y} f(\vec x) \overline{f(\vec y)} \,Z^{\vec p}\proj{\vec x+\vec q}Z^{-\vec p} \ \proj{\vec y}\\
=& \sum_{\vec x,\vec y} f(\vec x-\vec q) \overline{f(\vec y)}\, \ket{\vec x}\braket{\vec x|\vec y} \bra{\vec y}\\
=&\sum_{\vec x} f(\vec x-\vec q) \overline{f(\vec x)}\, \proj{\vec x}  
= B_{\partial_{-\vec q}f}\;.
\end{align*}
In the last equality we use the definition \eqref{ClassicalDerivative} to obtain the  claimed result. 
\end{proof}

\section{Quantum Polynomials and the Clifford Hierarchy}
We use the quantum derivative to give an inductive definition of  the exponential $Q^{(k)}$ of a quantum polynomial  of degree $k$.

\begin{Def}[\bf Quantum Polynomials]
Let the exponential $Q^{(0)}$ of a quantum polynomial of degree zero be a scalar multiple of the identity, $Q^{(0)}=\complex I$. A unitary $U\in Q^{(k)}$ is the exponential of a polynomial of degree $k$, if $\partial_{\vec a}U\in Q^{(k-1)}$ for all $\vec a\in V^n$.     
\end{Def}

The definition of quantum polynomials coincides with the Clifford hierarchy.

\begin{prop}[\bf Algebraic Hierarchy]\label{Prop:CliffordHierarchyAlgebraic}
For integer $k\geq 1$, the unitary exponentials of 
degree-k polynomials $Q^{(k)}$ equal the unitaries in the $k$-th level $\mathcal{C}^{(k)}$ of the Clifford hierarchy.
\end{prop}

\begin{lem}\label{lem:w}
If $U\in C^{(k)}$ and $w(\vec a), w(\vec b)$ are Weyl operators, then $w(a)Uw(b)\in C^{(k)}$.
\end{lem}

\begin{proof} 
By definition,  $C^{(1)}$ is the group generated by the the Weyl operators, so $U\in C^{(1)}$ ensures that $w(a)Uw(b)\in C^{(1)}$.  

We establish this by induction on $k$.
Assume the proposition is valid for $C^{(j)}$, with $j=1,\ldots,k$ and $U\in C^{(k+1)}$. Let $W=w(\vec a)Uw(\vec b)$. The induction is valid, if $W w(\vec c) W^{*}\in C^{(k)}$ for any $w(\vec c)$.
Note that 
    \begin{align*}
    W w(\vec c) W^{*}
    = \omega^{\lra{\vec b,\vec c}_{s}} \, w(\vec a) U w(\vec c)  U^*    w(a)^*
    = w(\vec a) V w(\vec a)^{*}\;,
\end{align*}
where $V=\omega^{\inner{ \vec b}{\vec c}_s}Uw(c)U^*\in C^{(k)}$. As we assume $U\in C^{(k+1)}$, it follows that $V\in C^{(k)}$.  So by the induction hypothesis, also $w(\vec a)V w(\vec a)^*\in C^{(k)}$.  Hence, $ w(\vec a)Uw(\vec b) \in C^{(k+1)}$.
\end{proof}

\begin{proof}[Proof of Proposition \ref{Prop:CliffordHierarchyAlgebraic}]
Let us start with case $k=1$. Any unitary $U\in Q^{(1)}$ and any Weyl operator $w(a)$ satisfy $\partial_{\vec a}U=w(\vec a)Uw(\vec a)^* U^{*}\in Q^{(0)}$. That is $w(a)Uw(a)^* U^*=cI$, where $c$ is a complex unit. Thus, $Uw(a)^* U=Uw(-a)U^*$ is a Weyl operator, from which we infer that $U\in C^{(1)}$.
Conversely, if  $U\in C^{(1)}$, then $U$ is a Weyl operator $w(b)$. Hence $w(a)Uw(a)^* U^* =w(a)w(b)w(a)^* w(b)^*=\omega^{\inner{a}{b}_s}_dI$, indicating $U\in Q^{(1)}$.

We proceed by induction. Assume that $C^{(j)}=Q^{(j)}$ for $j=1,\ldots,k$. Then if  $Uw(\vec a) U^*\in C^{(k)}$ for every $w(\vec a)$, one infers that $U\in C^{(k+1)}$. On the other hand, if $Uw(\vec a)^* U^*\in C^{(k)}$, then we infer from  Lemma \ref{lem:w} that  $w(\vec a)Uw(\vec a)^* U^*=\partial_{\vec a}U\in C^{(k)}=Q^{(k)}$.  Hence 
$U\in Q^{(k+1)}$.

Conversely, for $U\in Q^{(k+1)}$, we have $\partial_{\vec a}U=w(\vec a)Uw(\vec a)^* U^*=V\in Q^{(k)}=C^{(k)}$.  Again from Lemma \ref{lem:w}, we have $Uw(a)^* U^*=w(a)^* V\in C^{(k)}$. 
Hence, we have $U\in C^{(k+1)}$ and  $C^{(k+1)}=Q^{(k+1)}$.
\end{proof}

The 1st and 2nd levels of the Clifford hierarchy are closed under the adjoint operation;  this property does not hold in general.  However, some subsets  of unitaries are closed under the adjoint operation in any $k$-th level of Clifford hierarchy.
\begin{cor}\label{Cor:AdjointOK}
   The unitaries generated by degree-k polynomial
$U_f=\sum_{\vec x}\omega^{f(\vec x)}_d\proj{\vec x}$ are  closed under the adjoint operation in the $k$-th level of Clifford hierarchy.
 
\end{cor}

\section{Quantum Norms and Measures}\label{Sec:QuantumUniformityNorm}
Here we investigate in detail the quantum  measures mentioned in \eqref{eq:QUn}.  We use the term ``measure'' to designate a quantity that in the classical case is a norm; in this paper we show that these measures are norms for $k=2,3$. 

\begin{Def}
\label{Def:QGow}
Given a linear operator $B$ and an integer $k\geq 1$,
the quantum uniformity measure $\norm{\ \cdot \ }_{Q^k}$ of
order $k$ is defined by \eqref{eq:QUn}.
\end{Def}

We now prove the first inequality in Theorem \ref{Thm:Pos-Mono-Norm}. 
\begin{prop}[\bf Positivity]\label{Prop:Positivity}
The expression \eqref{eq:QUn} is non-negative, so the  quantum uniformity measure $\norm{\ \cdot \ }_{Q^k}$ can be taken as its positive root.  In particular for $k\geq2$,

\be
\norm{B}_{Q^k}^{2^k}
=\mathbb{E}_{\vec a_{k-1}, \ldots , \vec a_{1}}\abs{ \frac{1}{d^n} \Tr{\partial_{\vec a_{k-1},\cdots, \vec{a}_1} B } }^2\;,
\ee
and for $k=1$,
\be
\norm{B}^2_{Q^1} = \abs{\frac{1}{d^n}\Tr{B}}^{2}\;.
\ee

\end{prop}

\begin{proof} 
Write for $k\geq2$,
\begin{align}   \label{Qkmeasure} 
\norm{B}^{2^{k}}_{Q^k} 
&= \mathbb{E}_{\vec a_k, \vec a_{k-1}, \ldots , \vec a_{1}}\frac{1}{d^n}
\Tr{\partial_{\vec a_{k}}  \partial_{\vec a_{k-1},\cdots, \vec{a}_1} B} \nonumber\\
& \hskip -.35in 
=\mathbb{E}_{\vec a_{k-1}, \ldots , \vec a_{1}}\frac{1}{d^n}  
\Tr{\mathbb{E}_{\vec {a}} \, 
\lrp{\partial_{\vec a_{k-1},\cdots, \vec{a}_1} B} (\vec a)
\lrp{\partial_{\vec a_{k-1},\cdots, \vec{a}_1} B }^{*}}\nonumber\\
& \hskip -.35in =   \mathbb{E}_{\vec a_{k-1}, \ldots , \vec a_{1}}\abs{ \frac{1}{d^n} \Tr{\partial_{\vec a_{k-1},\cdots, \vec{a}_1} B } }^2\;.
\end{align}
Here we use \eqref{ErgodicIdentity} to obtain the final equality. In the case $k=1$, one omits the derivatives $\vec {a}_1, \ldots, \vec{a}_{k-1}$. As an expectation of an absolute squared quantity, the result is non-negative. See also \eqref{Q1-Positive}.
\end{proof}

\begin{prop}[\bf Inductive Property]\label{Prop:InductiveDefn}
The quantum uniformity measures $\norm{B}_{Q^{k}}$ satisfy the inductive relation,
\be
\norm{B}_{Q^k}^{2^k}
=\mathbb{E}_{\vec a}
\norm{\partial_{\vec a}B}^{2^{k-1}}_{Q^{k-1}}\;.
\ee
\end{prop}

\begin{proof}
Write
\begin{align}  
\norm{B}^{2^{k}}_{Q^k} 
= \mathbb{E}_{\vec {a}_{1}}\mathbb{E}_{\vec a_{k}, \ldots , \vec a_{2}}  \frac{1}{d^n}
\Tr{ \partial_{\vec a_{k},\cdots, \vec{a}_2} \partial_{\vec a_{1}}  B} 
\nonumber&=\mathbb{E}_{\vec a} \norm{\partial_{\vec a}B}^{2^{k-1}}_{Q^{k-1}}\;,
\end{align}
as claimed.
\end{proof}

\begin{prop}\label{lem:quan_inv}
The measure $\norm{\ \cdot\ }_{Q^k}$ is invariant under conjugation by a Clifford unitary. In other words,  for $U \in C^{(2)}$,
\begin{eqnarray}
\label{ConjugationInvariance}
\norm{UBU^{*}}_{Q^k}
=\norm{B}_{Q^k}\;.
\end{eqnarray}
For $k\geq 2$, the measure 
$\norm{\ \cdot\ }_{Q^k}$ is invariant under left or right  multiplication by a Weyl operator, 
\begin{eqnarray}\label{InvarianceUnderWeyl}
\norm{w(\vec a)B w(\vec b)}_{Q^k}
=\norm{B}_{Q^k}
\;.
\end{eqnarray}
\end{prop}

\begin{proof}
For $k=1$, we infer from Proposition \ref{Prop:Positivity}, that for any unitary $U$, 
\begin{align*}
\norm{UBU^*}_{Q^1}^2
= \abs{\frac{1}{d^n}\Tr{UBU^*}}^{2}
= \abs{\frac{1}{d^n}\Tr{B}}^{2}
=\norm{B}_{Q^1}^2\;.
\end{align*}
We proceed by induction on $k$. Assume hat the \eqref{ConjugationInvariance} holds  for $\norm{\ \cdot \ }_{Q^k}$, and consider $\norm{\ \cdot \ }_{Q^{k+1}}$. 
Rewrite the quantum derivative $\partial_{\vec a}(UAU^{*})$ as
\begin{align*}
\nonumber\partial_{\vec a}(UBU^{*})
&= w(\vec a) UBU^{*} w(\vec a)^*  UB^{*} U^{*}\\
&= UU^{*}w(\vec a) UBU^{*} w(\vec a)^*  UB^{*} U^{*} \\
&=U(\partial_{U^* w(\vec a) U}B) U^*.
\end{align*}
Then, using  Proposition \ref{Prop:InductiveDefn} and the inductive hypothesis, we have 
\begin{align*}
&\norm{UBU^*}_{Q^{k+1}}^{2^{k+1}}\\
=&\mathbb{E}_{\vec a}
\norm{\partial_{\vec a}(UBU^*)}_{Q^k}^{2^k}
=\mathbb{E}_{\vec a}
\norm{U(\partial_{U^* w(\vec a)U}B)U^*}_{Q^k}^{2^k}\\
=&\mathbb{E}_{\vec a}
\norm{\partial_{U^* w(\vec a)U}B}_{Q^k}^{2^k}
=\mathbb{E}_{\vec b}
\norm{\partial_{\vec b}B}_{Q^k}^{2^k}
=\norm{B}^{2^{k+1}}_{Q^{k+1}}\;.
\end{align*}
Here we use the fact that conjugation by $U^*$ maps $w(\vec a)$ to another element of the Weyl group, which equals $w(\vec b)$, up to a phase. Since this map in $\Z_d$ is 1-to-1, the summation over $\vec b$ gives the same result.

To establish \eqref{InvarianceUnderWeyl}, note that 
$w(\vec a) Bw(\vec b)=w(\vec a)Bw(\vec a)^* w(\vec a+\vec b)$ up to a  phase $\chi$, and $\norm{B\chi}_{Q^k}= \norm{B}_{Q^k} $ Since $w(\vec a)=U\in C^{(1)}$, the first statement of the proposition means that we only need to show that  $\norm{Bw(\vec b)}_{Q^k}
=\norm{B}_{Q^k}$.
The derivative $\partial_{\vec a}
(Bw(\vec b))$ can be written 
\begin{align}\label{PartialBw}
\partial_{\vec a} (Bw(\vec b))
= \omega^{\lra{\vec a,\vec b}_s}\, B(\vec a)B^*
= \omega^{\lra{\vec b,\vec a}_s}\, \partial_{\vec a} B\;.
\end{align}
Using Proposition \ref{Prop:InductiveDefn} and \eqref{PartialBw}, 
\begin{align*}
\norm{Bw(\vec b)}^{2^{k}}_{Q^{k}}
&=\mathbb{E}_{\vec a}
\norm{\partial_{\vec a}(Bw(\vec b))}^{2^{k-1}}_{Q^{k-1}}\\
&\hskip-.45in =\mathbb{E}_{\vec a}
\norm{\omega^{\inner{\vec b}{\vec a}_s}\partial_{\vec a}B}^{2^{k-1}}_{Q^{k-1}}
= \mathbb{E}_{\vec a}
\norm{\partial_{\vec a}B}^{2^{k-1}}_{Q^{k-1}}
=\norm{B}^{2^{k}}_{Q^{k}}.
\end{align*}
This completes the proof.
\end{proof}

\section{Schatten and Fourier Norms}
\label{Sect:SchattenFourier}
Robert Schatten introduced $L^{p}$ norms for linear transformations, which are used widely in analysis. We use the Schatten and Fourier norms on the coefficients $\Xi_{B}(\vec b)=\lra{w(\vec b), B}$ defined as
\begin{align*}
  \norm{B}_p=\lrp{\frac{1}{d^n} \Tr{(B^*B)^{p/2}}}^{1/p}\;,
  \ \ 
  \norm{\Xi_B}_{\ell^p}^{p} 
= { \sum_{\vec b}\abs{\Xi_{B}(\vec b)}^{p}   }\;.  
\end{align*}
For $p=2$ they are related by  
\be\label{Schatten-Fourier}
\norm{B}_{2} 
= \norm{\Xi_{B}}_{\ell^2}\;.
\ee

\section{Classical Uniformity Norms}
Our quantum uniformity measures $\norm{\ \cdot \ }_{Q^k}$ are a generalization of 
Gowers' classical uniformity norms $\norm{\ \cdot \ }_{U^k}$; see~\cite{Gowers98, tao2012higher}.  For a function $f(\vec x)$ on $\mathbb{Z}^n_d$, the classical uniformity norm can defined inductively as follows,
\begin{align*}
\norm{f}^{2^{k}}_{U^{k}}
= \mathbb{E}_{\vec q}
\norm{
\partial_{\vec q}f
}^{2^{k-1}}_{U^{k-1}}\;,
\ \text{and}\ 
\norm{f}_{U^{1}}
= \abs{\frac{1}{d^{n}} \sum_{\vec x \in \mathbb{Z}^n_d }  f(\vec x)}\;.
\end{align*}
For transformations $B$ that are diagonal in the computational basis, the quantum uniformity measures reduce to the classical norms.

\begin{prop}[\bf Reduction to Gowers' Norms] \label{Prop:ClassicalUniformity}
Let $B_f$ have the form \eqref{DiagonalOperator}, and let  $k\geq1$. Then 
\begin{eqnarray}\label{Reduction-Classical}
\norm{B_f}_{Q^k}
=\norm{f}_{U^{k}}\;.
\end{eqnarray}
\end{prop}
\begin{proof}
The identity \eqref{Reduction-Classical} holds for $k=1$ as a consequence of Proposition \ref{Prop:Positivity} and calculation of the trace in the computational basis.  
We proceed by induction on $k$.  
Assume \eqref{Reduction-Classical} holds for $j=1,\ldots, k-1$. Then using Proposition \ref{prop:red_der} to compute $\partial_{\vec a}B_f$, one has  
\begin{align*}
\norm{B_f}^{2^{k}}_{Q^{k}}
=\mathbb{E}_{\vec a}
\norm{\partial_{\vec a} B_f}^{2^{k-1}}_{Q^{k-1}}
=\mathbb{E}_{\vec a}
\norm{
B_{\partial_{-\vec q}f}
}^{2^{k-1}}_{Q^{k-1}}
=\norm{f}^{2^{k}}_{U^{k}}\;.
\end{align*}
In the last equality we use the inductive hypothesis.
\end{proof}

\section{Examples of Some Quantum Uniformity Measures}\label{Sect:SomeExamples}
Here we give the quantum uniformity measures for several gates, with  details in the Supporting Information.
We consider qudit systems with local dimension equal $2$ or an odd prime. Observe these quantum uniformity measures increase monotonically in  $k$, as proved later in Theorem~\ref{Thm:Monotonicity}.

\begin{Exam}[\bf Weyl Gates]\label{Exam:PauliGates}
    The Weyl (or Pauli) gates are given by the Weyl operators $w(\vec b)$. Their  quantum uniformity norms are 
\begin{align*}
    \norm{w(\vec b)}_{Q^1}=\delta_{\vec b, \vec 0}\;,\qquad
 \norm{w(\vec b)}_{Q^k}=1, \quad \forall k\geq 2\;. 
\end{align*}
   
\end{Exam}

\begin{Exam}[\bf The Fourier Gate]\label{Exam:FourierGate}
For $d=2$, this gate is called the  Hadamard gate. 
The Fourier gate for an $n$-qudit system of dimension $d$ is
$
F=\frac{1}{d^{n/2}}\sum_{x,y\in \Z_{d}^{n}}\omega^{\vec {x}\cdot \vec{y}}\ket{\vec {x}}\bra{\vec{y}}
$. 
\begin{align*}\label{FourierQ1}
        \norm{F}_{Q^1}
          &=\left\{
    \begin{aligned}
            \,0\;,\phantom{,} &\quad \text{if}~~ d=2\\
      \frac{1}{d^{n}}\;, &\quad \text{if d is odd prime}\\
    \end{aligned}
       \right.\;,\\
 \norm{F}_{Q^2} &=\left\{
    \begin{aligned}
        \frac{1}{2^{n/4}}, &\quad \text{if}~~ d=2\\
        \frac{1}{d^{n/2}}, &\quad \text{if d is odd prime}\\
    \end{aligned}
    \right.\;,\\
    \norm{F}_{Q^k}&=1, \quad \forall k\geq 3\;.
\end{align*}
\end{Exam}

\begin{Exam}[\bf The CNOT Gate]
The two-qudit ``control-not gate'' plays an important role in generating quantum entanglement. 
\be
CNOT=\sum_{x,y\in \mathbb{Z}_d}\proj{x}\ot \ket{x+y}\bra{y}\;.
\ee
For both $d=2$ or $d=$ odd prime, 
\begin{align*}
\norm{CNOT}_{Q^1}&=\frac{1}{d}\;,  \norm{CNOT}_{Q^2}=\frac{1}{\sqrt{d}},\\
\norm{CNOT}_{Q^k}&=1, \quad \forall k\geq 3\;.
\end{align*}

\end{Exam}

\begin{Exam}[\bf The CCZ Gate]
The three-qudit ``control-control-$Z$ gate,'' or $CCZ$, in the computational basis is  
\begin{align}
CCZ=\sum_{x,y,z}\omega^{xyz}_d\proj{x}\ot \proj{y}\ot\proj{z}.
\end{align}
The $CCZ$ gate also plays an important role in the theory of universal quantum computation. For $d$ an odd prime, 
\begin{align*}
&\norm{CCZ}_{Q^1}
=\frac{2d-1}{d^2}\\
&\norm{CCZ}_{Q^2}
=\left(\frac{d^3+(d-1)^3+d[d^3-(d-1)^3-1]}{d^6}  \right)^{1/4},\\
&\norm{CCZ}_{Q^3}=\left(\frac{d^3+3(d-1)d+3(d-1)^2d+(d-1)^3}{d^6}\right)^{1/8},\\
&\norm{CCZ}_{Q^k}=1, \forall k\geq 4
\end{align*}

If $d=2$, the corresponding quantum uniformity norms are  
\begin{align*}
&\norm{CCZ}_{Q^1}
=\frac{3}{4}, 
\norm{CCZ}_{Q^2}
=\left(\frac{11}{32}\right)^{1/4},\\
&\norm{CCZ}_{Q^3}=\left(\frac{11}{32}\right)^{1/8}, 
\norm{CCZ}_{Q^k}=1, \ \ \forall k\geq 4\;.
\end{align*}
\end{Exam}

\section{More Properties of the Quantum Uniformity Measures\label{Sect:MoreProperties}}
We explore more  mathematical properties of the quantum uniformity measures, including: monotonicity of 
$\norm{\cdot}_{Q^k}$ with respect to $k$, proving $\norm{\cdot}_{Q^k}$ is a norm for $k=2,3$, and showing relations between $\norm{\cdot}_{Q^k}$ and both Schatten and $\ell^{p}$ norms.

Given $2^k$ linear operators $\set{B_{\vec u}}_{\vec u\in \set{0,1}^k}$, with binary labels $\vec u$, we define 
\begin{align*}
\set{B_{\vec u}}_{\vec u\in \set{0,1}^k}(\vec a_k, .., \vec a_1)
\end{align*} 
inductively. 
For $k=1$,  
\begin{align*}
\set{B_{\vec u}}_{\vec u\in \set{0,1}^1}(\vec a_1)
=\set{B_0, B_1}(\vec a_1)
&:= B_0(\vec {a}_1) B_1^*
\;.
\end{align*}
This is chosen so that if $B_0=B_1=B$, then $\set{B_{\vec u}}_{\vec u\in \set{0,1}^1}(\vec a)
=\partial_{\vec a}B$. 
For $k=2$ we use the two intermediate sets
\begin{align*}
\set{B_{0\vec u}}_{\vec{u}\in\set{0,1}^1}(\vec{a}_1)
&= \set{B_{00},B_{01}}(\vec {a}_{1})\;,
\\
\set{B_{1\vec u}}_{\vec{u}\in\set{0,1}^1}(\vec{a}_1)
&= \set{B_{10},B_{11}}(\vec {a}_{1})\;,
\end{align*}
leading to 
\begin{align*}
\set{B_{\vec{u}}}_{\vec{u}\in\set{0,1}^{2}}
(\vec {a}_2, \vec{a}_1)
=&\set{B_{00}, B_{01}, B_{10}, B_{11}}(\vec a_2, \vec a_1)\\
=& w(\vec{a}_2)\, \lrp{\set{B_{0\vec{u}}}(\vec {a}_{1})}\,  w(\vec{a}_2)^*\,
\lrp{\set{B_{1\vec{u}}}(\vec {a}_{1})}^* \;.
\end{align*}

\begin{Def}
Given $\set{B_{\vec u}}_{\vec{u}\in(0,1)^1}(\vec{a}_1)$ above, for $k\geq1$ let   
\begin{align}
\set{B_{0\vec u}}_{\vec{u}\in\set{0,1}^{k}}
\lrp{\vec {a}_{k},\ldots,\vec {a}_{1}}
= \set{B_{0{00\cdots0}0},B_{0{00\cdots0}1}, \cdots, B_{0{11\cdots11}}}\lrp{\vec {a}_{k},\ldots,\vec {a}_{1}}\;,
\end{align}
\begin{align}
\set{B_{1\vec u}}_{\vec{u}\in\set{0,1}^{k}}
\lrp{\vec {a}_{k},\ldots,\vec {a}_{1}}= \set{B_{1{00\cdots0}0},B_{1{00\cdots0}1}, \cdots, B_{1{11\cdots11}}}\lrp{\vec {a}_{k},\ldots,\vec {a}_{1}}\;,
\end{align}
and  
\begin{align*}
\set{B_{\vec u}}_{\vec u\in \set{0,1}^{k+1}}(\vec a_{k+1}, .., \vec a_1)
:=&
w(\vec a_{k+1})\left(
\set{B_{0,\vec u}}_{\vec u\in \set{0,1}^{k}}(\vec a_{k}, .., \vec a_{1})
\right) w(\vec a_{k+1})^*\\
&\times 
\left(
\set{B_{1, \vec u}}_{\vec u\in \set{0,1}^{k}}(\vec a_{k}, .., \vec a_{1})
\right)^{*}.
\end{align*}
\end{Def}
The quantity $\set{B_{\vec u}}_{\vec u\in \set{0,1}^k}(\vec a_k, .., \vec a_1)$ is a generalization of the $k$-th order quantum derivative.  We use these quantities to define a quantum uniformity functional and a quantum uniformity inner product.

\begin{Def}[\bf Quantum Uniformity Functionals]
Given $\set{B_{\vec u}}_{\vec u\in \set{0,1}^{k+1}}$ with $2^{k+1}$ linear operators, the quantum uniformity inner functionals of $(k+1)^{\rm th}$ order are 
\begin{eqnarray}
\langle \set{B_{\vec u}}_{\vec u\in \set{0,1}^{k+1}}\rangle_{Q^{k+1}}
= \mathbb{E}_{\vec a_{k+1},...\vec a_1 \in V^n}
\frac{1}{d^n }\Tr{\set{B_{\vec u}}_{\vec u\in \set{0,1}^{k+1}}(\vec a_{k+1}, .., \vec a_1) }\;.
\end{eqnarray}
\end{Def}

If all the operators in the set $\set{B_{\vec u}}$ are equal $B$, then the quantum uniformity functional reduces to the quantum uniformity measures 
$\lra{\set{B_{\vec{u}}}_{\vec{u}\in \set{0,1}^{k}}}_{Q^{k}}=\norm{B}^{2^k}_{Q^k}$.

\begin{Def}\label{Def:InnerProduct}
Given two sets $\set{B_{0\vec{u}_k}}_{\vec u_{k}\in \set{0,1}^{k}}$  and $\set{B_{1\vec{u}_k}}_{\vec u_{k}\in \set{0,1}^{k}}$   of $2^k$ operators on $\mathcal{H}^{\otimes n}$, we regard them as subsets of $\set{B_{\vec{u}_{k+1}}}_{\vec u_{k+1}\in \set{0,1}^{k+1}}$.  The functional $\lra{\ \cdot \ }_{Q^{k+1}}$ determines a pairing $\lra{\ \cdot\  , \ \cdot\  }_{Q^{k}}$ of these subsets. 

For simplicity of notation, let us not explicitly note that $\vec{u}_{k} \in (0,1)^{k}$.  Then the pairing is 
\begin{align}\label{InnerProduct-2}
\lra{\set{B_{0\vec{u}_{k}}},   \set{B_{1\vec{u}_{k}}}}_{Q^{k}}
=
\lra{\set{B_{\vec{u}_{k+1}}}}_{Q^{k+1}}\;.
\end{align}
\end{Def}

\begin{prop}[\bf Inner Product]\label{Prop:PreInnerProduct}
With the notation introduced in Definition \ref{Def:InnerProduct}, the pairing \eqref{InnerProduct-2} is a pre-inner product. It is a non-negative, hermitian form that is linear in the first variable, conjugate linear in the second variable. It can be written   
\begin{align}\label{InnerProduct-3}
&\lra{\set{B_{0\vec{u}_{k}}},   \set{B_{1\vec{u}_{k}}}}_{Q^{k}}\nonumber =
\lra{\set{B_{\vec{u}_{k+1}}}}_{Q^{k+1}}\nonumber\\
&\quad\quad=
\underset{\vec a_1,...,\vec a_{k-1}}{\mathbb{E}}
\frac{1}{d^{2n} } \lrp{\Tr{\set{B_{0\vec {u}_{k}}}(\vec a_k, .., \vec a_1) }} \overline{\lrp{\Tr{\set{B_{1\vec {u}_{k}}}(\vec a_k, .., \vec a_1) }}} \;.
\end{align}
The pairing satisfies the Schwarz inequality,
\begin{align}
\abs{\lra{\set{B_{0\vec{u}}},    \set{B_{1\vec{u}}}}_{Q^{k}}}
\label{SchwarzInequality-Qk}\leq
\lra{\set{B_{0\vec{u}}},    \set{B_{0\vec{u}}}}^{\frac12}_{Q^{k}}
\lra{\set{B_{1\vec{u}}},    \set{B_{1\vec{u}}}}^{\frac12}_{Q^{k}}
\;.
\end{align}
\end{prop}

\begin{proof} 
Using  \eqref{ErgodicIdentity}, we have 
\begin{align}
\label{InnerProduct-6}&\lra{\set{B_{0\vec{u}}}_{u\in\set{0,1}^{k}},    \set{B_{1\vec{u}}}_{u\in\set{0,1}^{k}}}_{Q^{k}}
=\lra{\set{B_{\vec{u}_{k+1}}}}_{Q^{k+1}}\\
 \nonumber&=
\underset{\vec a_{k},...\vec a_1 \in V^n}{\mathbb{E}}
\frac{1}{d^{n} }{\Tr{\set{B_{0\vec u}}_{\vec u\in \set{0,1}^k}(\vec a_k, .., \vec a_1) }}{\overline{\frac{1}{d^{n} }\Tr{\set{B_{1\vec u}}_{\vec u\in \set{0,1}^k}(\vec a_k, .., \vec a_1) }}}\;.
\end{align}
From this representation, one can see the pairing has the correct linearity properties.  Furthermore, this shows that the pairing is hermitian, 
\be
\lra{\set{B_{0\vec{u}_{k}}},   \set{B_{1\vec{u}_{k}}}}_{Q^{k}}
= \overline{\lra{\set{B_{1\vec{u}_{k}}},   \set{B_{0\vec{u}_{k}}}}}_{Q^{k}}\;.
\ee
If $\set{B_{0\vec{u}}}=\set{B_{1\vec{u}}}$, then each term in the expectation is non-negative, so the expectation is non-negative.  To obtain \eqref{SchwarzInequality-Qk}, use the Schwarz inequality on the sequences in \eqref{InnerProduct-3}.
\end{proof}

Note that the expression \eqref{InnerProduct-6} is not strictly positive; for if  $k=1$, and $B_{00}=B_{10}=w(\vec b)$, $B_{01}=B_{11}=w(\vec c)$ with $\vec b\neq \vec c$, then \eqref{InnerProduct-6} is zero. In particular
\begin{align*}
0\leq \lra{\set{B_{0\vec{u}}}_{u\in(0,1)^{k}},    \set{B_{0\vec{u}}}_{u\in(0,1)^{k}}}_{Q^{k}} 
=&\mathbb{E}_{\vec {a}} \frac{1}{d^{2n}}
\abs{\Tr{w(\vec a) w(\vec b) w(\vec a)^* w(\vec c)^*}}^{2}\\
=& \delta_{\vec b, \vec c}=0\;.
\end{align*}
However, the form is strictly positive if all the $B$'s are the same.

Taking Proposition \ref{Prop:PreInnerProduct} into account, one can define
\begin{align}
\norm{{\set{B_{0\vec{u}}}_{u\in\set{0,1}^{k}}}}_{Q^{k}}
=
\lra{\set{B_{0\vec{u}}}_{u\in\set{0,1}^{k}},    \set{B_{0\vec{u}}}_{u\in(0,1)^{k}}}_{Q^{k}}^{\frac{1}{2}}\;,
\end{align}
as the right-hand side is non-negative. Similarly 
\begin{align}
\norm{{\set{B_{1\vec{u}}}_{u\in\set{0,1}^{k}}}}_{Q^{k}}
=
\lra{\set{B_{1\vec{u}}}_{u\in\set{0,1}^{k}},    \set{B_{1\vec{u}}}_{u\in\set{0,1}^{k}}}_{Q^{k}}^{\frac{1}{2}}\;.
\end{align}

\begin{thm}[\bf Monotonocity of Measures] \label{Thm:Monotonicity}
For any $n$-qudit linear operator $B\in L(\mathcal{H}^{\ot n})$,  the measures \eqref{eq:QUn} satisfy
\begin{eqnarray}
\norm{B}_{Q^k}\leq 
\norm{B}_{Q^{k+1}},\quad \text{for all } k\geq 1.
\end{eqnarray}
\end{thm}
\begin{proof}
Choose $\set{B_{0}}_{\vec{u}\in\set{0,1}^1}= \set{B}$ and $\set{B_{1}}_{\vec{u}\in\set{0,1}^1}= \set{I}$.  Then Proposition \ref{Prop:PreInnerProduct} yields the desired bound.
\end{proof}

We now  illustrate some further properties of $\norm{\cdot}_{Q^k}$ for $k=2,3$.

\begin{prop}\label{Prop:Q2Form}
For $k=2$ one has 
\begin{align*}
\lra{B_{00}, B_{01}, B_{10}, B_{11} }_{Q^{2}}
&= \sum_{\vec b} \Xi_{B_{00}}(\vec b)\,  \overline{\Xi_{B_{01}}(\vec b)} \, \overline{ \Xi_{B_{10}} (\vec b)}\Xi_{B_{11}} (\vec b) \;.\nonumber
\\
\abs{\lra{B_{00}, B_{01}, B_{10}, B_{11} }_{Q^{2}}}
&\leqslant \prod_{i,j=0}^1\norm{B_{ij}}_{Q^2}\;.
\end{align*}
Also $\norm{B}_{Q^{2}}=\norm{\Xi_{B}}_{\ell^4}$, and  $\norm{\ \cdot\ }_{Q^{2}}$ is a norm. 
\end{prop}

\begin{proof}
As $T^{*}=\sum_{\vec b} \overline{ \Xi_{T}(\vec b)} \,w(-\vec b)$, it follows that 
\begin{align*}
&\lra{B_{00}, B_{01}, B_{10}, B_{11} }_{Q^{2}}\\
&\qquad = \mbe_{\vec a_{1}, \vec a_{2}} \frac{1}{d^{n}}
\Tr{\lrp{
B_{00}(\vec a_{1}) B_{01}^{*} }(\vec a_{2}) \lrp{B_{10}(\vec a_{1}) B_{11} ^{*}  }^{*}    }\nonumber\\
&\qquad = \mbe_{\vec a_{1}, \vec a_{2}} \frac{1}{d^{n}} 
\sum_{\vec b_{1}, \vec b_{2}, \vec b_{3}, \vec b_{4}}
 \omega^{\lra{ \vec a_{1}+\vec a_{2}     , \vec b_{1}      }_{s}  -\lra{\vec a_{2}     , \vec b_{2}      }_{s}  -\lra{\vec a_{1}, \vec b_{3}}_{s} }\\
& \qquad\qquad \times 
\Xi_{B_{00}}(\vec b_{1})  \overline{\Xi_{B_{01}}(\vec b_{2}) }\,  \overline{  \Xi_{B_{10}}(\vec b_{3}) }  \Xi_{B_{11}}(\vec b_{4})
\Tr{w(\vec b_{1})  w(-\vec b_{2})  w(-\vec b_{3})   w(\vec b_{4})}
\\
&\qquad =  
\sum_{\vec b_{1}, \vec b_{2}, \vec b_{3}, \vec b_{4}}
\delta_{\vec b_{1}, \vec b_{3}}     \delta_{\vec b_{1}, \vec b_{2}}  \delta_{\vec b_{1}-\vec b_{2}, \vec b_{3} -\vec b_{4}   } 
\Xi_{B_{00}}(\vec b_{1})  \overline{\Xi_{B_{01}}(\vec b_{2}) }\,  \overline{  \Xi_{B_{10}}(\vec b_{3}) }  \Xi_{B_{11}}(\vec b_{4})
\\
&\qquad =  
\sum_{\vec b}
\Xi_{B_{00}}(\vec b)  \overline{\Xi_{B_{01}}(\vec b) } \, \overline{  \Xi_{B_{10}}(\vec b) }  \Xi_{B_{11}}(\vec b)\;.
\end{align*}
From this identity and  H\"older's inequality, we infer the upper bound. 
Setting the $B$'s equal, one has $\norm{B}_{Q^{2}}=\norm{\Xi_{B}}_{\ell^4}$.  Since $\Xi_{B}$ is linear in $B$ and $\norm{\Xi_{B}}_{\ell^4}$ is a norm, also $\norm{B}_{Q^{2}}$ is a norm.  
\end{proof}

For the case where $k=3$, the relation to $\ell^p$ norms in Fourier space is more complicated;  one has several independent vector summations, rather than only 1 for $k=2$.  In order to obtain a generalized Schwarz inequality, we establish the following.

\begin{lem}\label{lem:Q3}
    Given 8 linear operators $\set{B_{\vec u}}_{\vec u\in\set{0,1}^3}$, we have two permutation identities:
     \begin{align}
        \nonumber&\langle B_{000}, B_{001}, B_{010}, B_{011}, B_{100}, B_{101}, B_{110}, B_{111}\rangle_{Q^3}\\
       \label{eq:Q3_1} 
       =&\langle B^{*}_{010}, B^{*}_{000}, B^{*}_{011}, B^{*}_{001}, B^*_{110}, B^*_{100}, B^*_{111}, B^*_{101}\rangle_{Q^3},
    \end{align}
    and 
   \begin{align*}
      \nonumber & \langle B_{000}, B_{001}, B_{010}, B_{011}, B_{100}, B_{101}, B_{110}, B_{111} \rangle_{Q^3}\\
     \label{eq:Q3_2}   
     =&\langle B_{011}, B_{010}, B_{111}, B_{110}, B_{001}, B_{000}, B_{101}, B_{100}\rangle_{Q^3}.
    \end{align*}
\end{lem}

\begin{proof}
To simplify notation,  denote $T(\vec {a}_1)$ by $T(1)$.  To establish the first equality, use  Proposition \ref{Prop:PreInnerProduct} to obtain
\begin{align*}
&\langle{B_{000}, B_{001}, B_{010}, B_{011}, B_{100}, B_{101}, B_{110}, B_{111}}\rangle_{Q^3}\\
=&\mathbb{E}_{2, 1}\frac{1}{d^{2n}}\Tr{B_{000}(1+2) B^{*}_{001}(2) B_{011} B^{*}_{010}(1)}
\overline{\Tr{B_{100}(1+2)B^{*}_{101}(2)
 B_{111} B^*_{110}(1) 
 }}\;.
\end{align*}
Use cyclicity of the trace and translation by $-\vec {a}_2$ to  rewrite the first trace as
\[
\Tr{B^{*}_{010}(1-2)B_{000}(1) B^{*}_{001} B_{011}(-2) }\;.
\] 
Make the same manipulation of the second trace. The expectation over $\vec a_2$ equals the expectation over $-\vec a_2$. Also one can interchange the expectation over $\vec a_1$ and $\vec a_2$. This gives the desired identity.
In order to establish the second equality, note that 
\begin{align*}
&\langle{B_{000}, B_{001}, B_{010}, B_{011}, B_{100}, B_{101}, B_{110}, B_{111}}\rangle_{Q^3}\\
=&\mathbb{E}_{3,2,1}\frac{1}{d^n}
\text{Tr}\left\{
B_{000}(1+2+3)B^*_{001}(2+3)  B_{011}(3)
 B^*_{010}(1+3)\right.\\\
& \qquad\qquad\qquad \left.\times B_{110}(1)  B^*_{111}   B_{101}(2) 
 B_{100}(1+2) \right\}\;.
 \end{align*}
One completes the proof in the same manner as was used to establish the first identity.
\end{proof}

\begin{prop}[\bf Generalized Quantum  Schwarz Inequality] \label{Prop:GeneralizedSchwarz}
Let $k=3$ and let $ \set{B_{\vec u}}_{\vec u\in \set{0,1}^{k}}$ be $8$ linear operators.  Then,
\begin{eqnarray}
|\langle \set{B_{\vec u}}_{\vec u\in \set{0,1}^{k}}\rangle_{Q^{k}}|
\leq \Pi_{\vec u\in \set{0,1}^k}\norm{B_{\vec u}}_{Q^k}\;.
\end{eqnarray}
\end{prop}

The proof of Proposition \ref{Prop:GeneralizedSchwarz} is in Appendix A.

\begin{thm}
For $k=3$, the $\norm{\ \cdot\ }_{Q^k}$ defines a norm.
\end{thm}

\begin{proof}
A norm has four properties: positivity, scaling, non-degeneracy, and the triangle inequality. Positivity for all $k$ was established in Proposition \ref{Prop:Positivity}. Scaling for all $k$  is clear from the definition. Non-degeneracy is true as $k=2$ gives a norm and 
the monotonicity property established in Theorem \ref{Thm:Monotonicity}. 
Thus we only need to prove the triangle inequality, which we prove for $k=3$.
\begin{align*}
        \norm{B_0+B_1}^{2^k}_{Q^k}
        &=\langle\set{B_0+B_1}_{\vec u\in\set{0,1}^k}\rangle_{Q^k}\\
        &=\sum_{S\subset\set{0,1}^k}
        \langle\set{B_{I_{\vec u\in S}}}_{\vec u\in\set{0,1}^k}\rangle_{Q^k}\\
        &\leq\sum_{S\subset\set{0,1}^k}
        \Pi_{\vec u\in \set{0,1}^n}\norm{B_{I_{\vec u\in S}}}_{Q^k}\\
            &=\left(\norm{B_0}_{Q^k} +\norm{B_1}_{Q^k}\right)^{2^k}\;.
    \end{align*}
Here $I_{\vec u\in S}$ is the indicator function for the subset $S$ that chooses $B_0$ in the set $S$ and $B_1$ in its complement. We have used the upper bound in Proposition \ref{Prop:GeneralizedSchwarz} for each term with $k=3$.
Hence, we have the triangle inequality $\norm{B_0+B_1}_{Q^k}\leq \norm{B_0}_{Q^k}+\norm{B_1}_{Q^k}$ in case that  $k=3$.
\end{proof}

\begin{Rem}
It would be interesting to show whether  $\norm{\ \cdot\ }_{Q^k}$ satisfies the triangle inequality for $k\geq 4$.
\end{Rem}

In the classical theory of uniformity norms, the connection between $L^p$ norms
and Gowers' norms has been studied by the Eisner and Tao~\cite{TT12}.
In this work, we investigate the connection between the quantum uniformity measure $\norm{\ \cdot \ }_{Q^k}$  (or norm with $k=2,3$) and the $L^{p}$ norm. 

\begin{lem}\label{230810lem1}   
Given two positive operators $A,B\in L(\mathcal{H}^{\ot n})$ and integer $k\geq 0$, one has   
\begin{align}
     0\leq\mathbb{E}_{\vec{a}\in V^n} \frac{1}{d^n} \Tr{ \lrp{A(\vec{a})\, B}^k }
     \leq  \|A\|_{k}^k \|B\|_{k}^k.
\end{align}
\end{lem}

\begin{proof}
Note that $(AB)^{k}=AB^{1/2} \lrp{B^{1/2}AB^{1/2}}^{k-1} B^{1/2}$, so 
\[
\Tr{(AB)^k}=\Tr{ \lrp{B^{1/2}AB^{1/2}}^{k} }\geq0\;.
\]
The Araki–Lieb–Thirring inequality \cite{LiebThir76,Araki90} states that for positive operators $A,B$ and positive integer $k$, 
\begin{align*}
    0\leq\Tr{(AB)^{k}}
    \leq
    \Tr{A^{k}B^{k}}\;.
\end{align*}
As $A\mapsto A(\vec{a})$ preserves positivity,
\begin{align*}
0\leq\mathbb{E}_{\vec{a}\in V^n} \frac{1}{d^n} \Tr{ ( A (\vec{a}) B)^k }  
&\leq \mathbb{E}_{\vec{a}\in V^n} \frac{1}{d^n} \Tr{ A^k(\vec{a}) B^{k }  }
\;.
\end{align*}
Using \eqref{ErgodicIdentity}, one has
\begin{align*}
\mathbb{E}_{\vec a} \frac{1}{d^n} \Tr{ A^k(\vec{a}) B^{k }  }
&=
\left(\frac{1}{d^n}\Tr{A^k}\right)\left(\frac{1}{d^n}\Tr{B^k}\right)=\|A\|_{k}^k \|B\|_{k}^k\;.
\end{align*}
Therefore, we obtain the result.
\end{proof}

\begin{prop}\label{230810prop1}
Given two operators $A,B\in L(\mathcal{H}^{\ot n})$ and an integer $k\ge 1$,
we have
\begin{align}
\label{230810eq1}
\left( \mathbb{E}_{\vec a\in V^n} \| A(\vec a) B\|_r^s\right)^{1/s} \le \|A\|_p \|B\|_q,
\end{align}
where $p = q = {2^k}/(k + 1)$, $r = {2^{k-1}}/k$, and $s = 2^{k-1}$.
\end{prop}
We prove Proposition \ref{230810prop1} in  Appendix A.

\begin{thm}[\bf Schatten Norms Dominate Quantum Uniformity]\label{thm:rel_lp}
For integer $k\ge 1$ and  $B \in L(\mathcal{H}^{\ot n})$, 
\[\|B\|_{Q^k} \le \|B\|_{p_k}
\;,\quad\text{where }p_{k}=\frac{2^{k}}{k+1}\;.
\]
\end{thm}
\begin{proof}
When $k=1$, we use \eqref{ErgodicIdentity} to infer that 
\[\|B\|_{Q^1}^2= \mathbb{E}_{\vec a\in V^n} \frac1{d^n} \Tr{ B(\vec a)B^*} = \left|\frac1{d^n} \Tr{B}\right|^2 \le \|B\|_1^2 . \] 
Here we also use the well-known bound  $\abs{\frac{1}{d^n}\Tr{B}}\leq \norm{B}_{1}$.

For $k>1$ we proceed inductively. Assume the statement holds for $k=1,\ldots, K$.
For the $k+1$ case, we have 
\begin{align*}
\|A\|_{Q^{k+1}}^{2^{k+1}} = \mathbb{E}_{ \vec a \in V^n} \| \partial_{\vec a} A\|_{Q^{k}}^{2^k} \le \mathbb{E}_{ \vec a \in V^n} \| A(\vec a) A^\dag\|_{p_k}^{2^k}
\le   \|A\|_{p_{k+1}}^{2^{k+1}},
\end{align*}
where the first inequality used the inductive assumption, and the second used Proposition \ref{230810prop1}.
\end{proof}

\begin{prop}[\bf Pure States]\label{Prop:Q34Pure}
Consider any $n$-qudit, pure state $\rho=\proj{\psi}=\sum_{\vec a\in V^n}  \Xi_{\rho}(\vec a)\,  w(\vec  a)$. Then
\begin{align*}
   \norm{\rho}_{Q^{1}}=&\frac{1}{d^n}; \quad \norm{\rho}_{Q^{2}}=\norm{\Xi_\rho}_{\ell^4};\\
\norm{\rho}_{Q^{3}}
=&d^{n/4} \norm{\Xi_\rho}_{\ell^4}; \quad
\norm{\rho}_{Q^{4}}=d^{n/2}\norm{\Xi_\rho}^{1/2}_{\ell^8}\norm{\Xi_\rho}^{1/2}_{\ell^4}.
\end{align*}
In general for $k\geq3$,  the $Q^{k}$ measure of a pure state $\rho$ is given by $\ell^{2^{p}}$ norms of its Fourier coefficients $\Xi_{\rho}(\vec a)$ for $2\leq p\leq k-1$. 
\be\label{PureStateQMeasures}
\norm{\rho}_{Q^{k}} = d^{n\lrp{1-\frac{2k}{2^{k}}}   } \norm{\Xi_{\rho}}_{\ell^{4}}^{2^{-k+2}}
\prod_{j=0}^{k-3} \norm{\Xi_{\rho}}_{\ell^{2^{(j+2)}}}^{2^{-k+j+2}}\;.
\ee
\end{prop}

\begin{proof}
For $k=1,2$ use Proposition \ref{Prop:NormReductions}. 
For $k\geq3$,
we have 
\begin{align*}
    \norm{\rho}^{2^k}_{Q^k}=&
\mathbb{E}_{\vec a}\norm{\partial_{\vec a}\rho}^{2^{k-1}}_{Q^{k-1}}\\
=&\mathbb{E}_{\vec a}|\bra{\psi} w(\vec a)^*\ket{\psi}|^{2^{k-1}}\norm{w(\vec a)\rho}^{2^{k-1}}_{Q^{k-1}},\\
=&\mathbb{E}_{\vec a}d^{n2^{k-1}}|\Xi_{\rho}(\vec a)|^{2^{k-1}}\norm{w(\vec a)\rho}^{2^{k-1}}_{Q^{k-1}},\\
=&\mathbb{E}_{\vec a}d^{n2^{k-1}}|\Xi_{\rho}(\vec a)|^{2^{k-1}}\norm{\rho}^{2^{k-1}}_{Q^{k-1}},\\
=&d^{n(2^{k-1}-2)} \norm{\Xi_{\rho}}^{2^{k-1}}_{l^{2^{k-1}}}\norm{\rho}^{2^{k-1}}_{Q^{k-1}},
\end{align*}
where the first line comes from Proposition \ref{Prop:InductiveDefn}.
The second line comes from the fact that
$\partial_{\vec a} \rho=\bra{\psi} w(\vec a)^*\ket{\psi}w(\vec a)\rho$ for pure state $\rho=\proj{\psi}$.
The third line comes from the fact that $\bra{\psi} w(\vec a)^*\ket{\psi} = d^{n} \,\Xi_{\rho}(\vec a)$,
and the forth line comes \eqref{InvarianceUnderWeyl} in Proposition \ref{lem:quan_inv}.
By using the inductive relation $\norm{\rho}^{2^k}_{Q^k}=d^{n(2^{k-1}-2)} \norm{\Xi_{\rho}}^{2^{k-1}}_{l^{2^{k-1}}}\norm{\rho}^{2^{k-1}}_{Q^{k-1}}$, we get the result.
\end{proof}

\begin{Rem}
An $n$-qudit unit vector, $\ket{\psi}=\sum_{\vec x} f_{\psi}(\vec x)\ket{\vec x}$ in the computational basis, gives the pure state $\rho_{\psi}=\proj{\psi}=\sum_{\vec a}\Xi_{\rho_{\psi}}(\vec a)w(\vec a)$. In case $\rho_{\psi}$ is diagonal in the computational basis, i.e., $\rho_\psi=\proj{\vec x}$, we infer from Propositions \ref{Prop:ClassicalUniformity} and \ref{Prop:Q34Pure}
that $\norm{\rho_{\psi}}_{Q^k}=\norm{f_{\psi}}_{U^k}=d^{-n\frac{k+1}{2^k}}$. But generally $\norm{\rho_\psi}_{Q^{k}}\neq \norm{f_\psi}_{U^k}$. 
In the simplest case $k=1$, using Proposition \ref{Prop:NormReductions}, $\norm{\rho_{\psi}}_{Q^1}=\frac{1}{d^{n}}$, while $\norm{f_{\psi}}_{U^1}=|\mathbb{E}_{\vec x}f_{\psi}(\vec x)|$, which  depends on the  vector $\ket{\psi}$.
\end{Rem}
Next we show logarithmic convexity of the quantum uniformity measures in the special case of pure states and certain $k$. It would be interesting to generalize this to mixed states and arbitrary $k$.

\begin{prop}[\bf Log-convexity of the quantum uniformity norm]\label{thm:log_conv}
Let $\rho$ be a pure state in $L(\mathcal{H}^{\ot n})$ for integer local dimension $d\ge 2$.  Then \begin{align}\label{230807eq1-}
\|\rho\|_{Q^2}\le \|\rho\|^{1/2}_{Q^1}\|\rho\|^{1/2}_{Q^3}.
\end{align}
Moreover, equality holds in \eqref{230807eq1-}, if and only if $\rho$ is a stabilizer state.
\end{prop}

\begin{proof}
From Proposition \ref{Prop:Q2Form} we infer that  $\norm{\rho}_{Q^2}
=\norm{\Xi_{\rho}}_4$. And from Proposition \ref{Prop:ExpectationDerivative} and $\Tr{\rho}=1$ we infer that $\norm{\rho}_{Q^1}^2 =\frac{1}{d^{2n}}$.
Therefore  \eqref{230807eq1-} is equivalent to 
\begin{align}\label{TheThingToShow-1}
    d^{n}  \norm{\Xi_{\rho}}_{\ell^4}^{2}
    \leq \norm{\rho}_{Q^{3}}\;.
\end{align}
Inserting the value of $\norm{\rho}_{Q^3}$ found in  Proposition \ref{Prop:Q34Pure}, we find that one  needs to verify that 
\begin{align}\label{TheThingToShow}
    d^{n}  \norm{\Xi_{\rho}}_{\ell^4}^{2}
    \leq d^{n/4}\norm{\Xi_\rho}_{\ell^4}\;,
    \quad\text{or}\quad
    \norm{\Xi_\rho}_4 \leq d^{-3n/4}\;.
\end{align}

Since   ${\Xi_{\rho}(\vec a)}
= {\lra{w(\vec {a}),\rho}}
=
\frac{1}{d^n} {\Tr{w(\vec a)^*\rho}}$, one infers from $\abs{\Tr{w(\vec a)\rho}}\leq \Tr{\rho}=1$,   
\[
\abs{\Xi_{\rho}(\vec a)}
\leq \frac{1}{d^n}\;.
\]
Also $\rho$ is pure, so $ \Tr{\rho^2}=1$,
and $
\sum_{\vec a\in V^n} |\Xi_\rho(\vec a)|^2= \frac{1}{d^n}$. Then 
\be
\sum_{\vec a} \abs{\Xi_{\rho}(\vec {a})}^{4}
\leq
 \max_{\vec{b}}\abs{\Xi_{\rho}(\vec {b})}^{2}\sum_{\vec a}\abs{\Xi_{\rho}(\vec {a})}^{2}
\leq \frac{1}{d^{2n}} \frac{1}{d^n}\;,
\ee
so  \eqref{TheThingToShow}  holds.

Equality pertains if and only if $|\Xi_\rho(\vec a)|\in \left\{0,\frac{1}{d^n}\right\}$ for every $\vec a$,
which means the pure state $\rho$ is a stabilizer state. This characterization of stabilizers follows from  
Proposition \ref{prop:stab_char}.
\end{proof}

\section{Characterization by Quantum Convolution}\label{Sect:Convolution}
In this section we give an alternative way to characterize the quantum uniformity measures using the Hadamard convolution $\boxtimes_H$, previously introduced by the authors in~\cite{BGJ23a,BGJ23b,BGJ23c}.  In the following we define convolution by starting from the tensor product of 2 copies of the $n$-qudit Hilbert space $\mathcal{H}^{\otimes n}$, which we denote as  $\mathcal{H}_{I}$ and $\mathcal{H}_{II}$. In considering $\mathcal{H}_{I}\otimes \mathcal{H}_{II}$, we denote $\Ptr{II}{\ \cdot\ }$ as the partial trace over $\mathcal{H}_{II}$, bringing one back to the Hilbert space $\mathcal{H}_{I}=\mathcal{H}^{\otimes n}$.

\begin{Def}[\bf Hadamard Convolution]\label{Def:Had_Conv}
The Hadamard convolution  of two $n$-qudit states $\rho$ and $\sigma$ is the state
\begin{eqnarray}
\rho\boxtimes_H\sigma=\Ptr{II}{V_H(\rho\ot\sigma) V^*_H}\;.
\end{eqnarray}
In this expression we choose  $V_H=U^{\ot n}_H:=U^{(1,n+1)}_H\ot U^{(2,n+2)}_H\ot...\ot U^{(n,2n)}_H$. Here  
$U_H$ is the $2$-qudit Hadamard unitary 
\begin{eqnarray}
U_H=\sum_{x,y\in\mathbb{Z}_d}
\ket{x}\bra{ x+ y}
\ot \ket{y}\bra{ x- y}\;,
\end{eqnarray}
and  $U^{(i,n+j)}_H$ denotes the action of $U_H$ on the $i^{\rm th}$ qudit in $\mathcal{H}_{I}$  and the $j^{\rm th}$ qudit in $\mathcal{H}_{II}$. 
The corresponding convolutional channel $\mathcal{E}_H$ is 
\begin{eqnarray}
\mathcal{E}_H(\ \cdot\ )
=\Ptr{II}{V_H \ \cdot\  V^*_H}\;.
\end{eqnarray}
\end{Def}

\begin{lem}\label{Prop:ConvolutionHad-w}
Let $B=\sum_{\vec a}\Xi(\vec a)w(\vec a)$ be an $n$-qudit transformation of local dimension $d$. Then 
\be
\norm{B\boxtimes_{H} B }_{2}^{1/2} = d^{n/2}\lrp{ \sum_{\vec a} \abs{\Xi_{B}(\vec a)}^{4}}^{1/4}=d^{n/2}\norm{\Xi_B}_4\;.
\ee
\end{lem}

 \begin{proof}
  This is a direct consequence of Proposition 75 in \cite{BGJ23b}, noting that $\Xi$ is normalized differently in that work. In more detail, 
the Hadamard convolution of two Weyl operators satisfies:
\be\label{wHConvolution}
w(\vec a) \boxtimes_{H} w(\vec b) 
= d^{n} \delta_{\vec a, \vec b} \,\omega^{-\lra{\vec a_{1}\vec a_{2}}} Z^{\vec a_{1}} w(\vec a)
\;.
\ee
Thus $B\boxtimes_{H}B= \sum_{\vec a} d^{n}\Xi_{B}(\vec a)^{2} w'(\vec a)$, where $w'(\vec a)=Z^{\vec a_{1}} w(\vec a)$, up to a phase, is an orthonormal basis in the inner product determined by the Schatten 2-norm. In fact \eqref{Schatten-Fourier} shows the Schatten 2-norm squared of $B$ equals the  $\ell^{2}$ norm of $\Xi_{B\boxtimes_{H}B}$. This establishes the proposition as claimed.
\end{proof}

\begin{prop}\label{prop:con_unif}
For  $B\in L(\mathcal{H}^{\ot n})$, we have both
\begin{align*}
&d^{n/2}\norm{B}_{Q^2} 
= 
\norm{B\boxtimes_H B}^{1/2}_{2}\;,\\
  &
  d^{n/2}\norm{B}^{2^{k+1}}_{Q^{k+1}}
    =\underset{\vec a_1,...,\vec a_{k-1}}{\mathbb{E}}\norm{(\partial_{\vec a_{k-1}, \vec a_{k-2},...,\vec a_1} B )\boxtimes_H (\partial_{\vec a_{k-1}, \vec a_{k-2},...,\vec a_1} B)}^2_{2}\;.
\end{align*}
\end{prop}
\begin{proof}
In Proposition \ref{Prop:Q2Form} we have shown that $\norm{B}_{Q^2} =\norm{\Xi_B}_{\ell^4}$.  In Proposition \ref{Prop:ConvolutionHad-w} we show that  $\norm{B\boxtimes_H B}^{1/2}_{2}= d^{n/2}   \norm{\Xi_B}_{\ell^4}$. 
In addition, by Proposition \ref{Prop:InductiveDefn}, one infers   
$ \norm{B}^{2^{k+1}}_{Q^{k+1}}$ satisfies 
\begin{align}
    \norm{B}^{2^{k+1}}_{Q^{k+1}}
    =\underset{\vec a_1,...,\vec a_{k-1}}{\mathbb{E}}\norm{\partial_{\vec a_{k-1}, \vec a_{k-2},...,\vec a_1} B }^4_{Q^{2}}\;,
\end{align}
as claimed.
\end{proof}

\section{Overlap with the Clifford Hierarchy}  
We have shown that the measures $\norm{\ \cdot\ }_{Q^{k+1}}$ characterize exactly whether a unitary $U$ is an element of the $k^{\rm th}$-level $\mathcal{C}^{(k)}$ of the Clifford hierarchy.  Here we introduce a different measure $\norm{U}_{q^{k+1}}$, which we use to  bound  the $L^2$-Schatten distance between a given unitary $U$ and $\mathcal{C}^{(k)}$. In classical HOFA the analogs of such measures  have been extensively studied. 
\begin{Def}
The overlap of a unitary $U$ with the $k^{\rm th}$ level $\mathcal{C}^{(k)}$  of the  Clifford hierarchy is 
\begin{eqnarray}
\norm{U}_{q^{k+1}}
=\max_{V\in C^{(k)}}
|\inner{V}{U}|^2.
\end{eqnarray}
As $C^{(k)}\subset C^{(k+1)}$, then we 
have $\norm{U}_{q^k}\leq \norm{U}_{q^{k+1}}$ for any $k$.
\end{Def}

\begin{prop}\label{Prop:AnalyticCliffordCharacterization}
    Given an $n$-qudit unitary $U$ and $k\geq1$, one has 
    \be
    \norm{U}_{q^{k}}=1 \quad \text{iff} 
    \quad
    \norm{U}_{Q^{k}}=1\;.
    \ee
    In other words, $U$ is in the $k^{\rm th}$ level of Clifford hierarchy iff 
    $\norm{U}_{Q^{k+1}}=1$
\end{prop}

\begin{proof}
By the inductive definition, we only need to prove the case where $k=1$. That, $U$ is equal to identity up to some phase, iff 
$\norm{U}_{Q^{1}}=1$.
This is because $\norm{U}_{Q^{1}}^2=|\frac{1}{d^n}\Tr{U}|^2=1$
 iff $U$ is equal to identity up to some phase.
\end{proof}

\begin{lem}\label{Lem:FirstQuartic}
For any two  $n$-qudit unitaries $U, V$, 
\begin{eqnarray}\label{LowerBoundQ2}
\mathbb{E}_{\vec a}
\abs{\inner{\partial_{\vec a}U}{\partial_{\vec a} V}}^2
= \sum_{\vec a} \abs{\Xi_{U^{*}V}(\vec a)}^{4}\;.
\end{eqnarray}
\end{lem}

\begin{proof}
This is a consequence of \eqref{TraceDerivativeProduct} and 
Proposition \ref{Prop:Q2Form}, with $B=U^{*}V$.
\end{proof}

\begin{thm}[\bf Overlap and the Uniformity Measures]
Given an $n$-qudit unitary $U$,  and  $k\geq 1$, 
\begin{eqnarray}
\norm{U}_{q^k}
\leq \norm{U}_{Q^k}\;.
\end{eqnarray}
\end{thm}

\begin{proof}
First, let us prove that for any unitaries $U$ and $V$,
\begin{eqnarray}\label{LowerBoundQ2-1}
\mathbb{E}_{\vec a}
|\inner{\partial_{\vec a}U}{\partial_{\vec a} V}|^2
\geq |\inner{U}{V}|^4\;.
\end{eqnarray}
 In fact by Lemma \ref{Lem:FirstQuartic} and 
$
\lra{\partial_{\vec{a}}U, \partial_{\vec{a}}V}
= \frac{1}{d^n} \Tr{\partial_{\vec a}(U^{*}V)}
$,
\begin{align*}
    \mathbb{E}_{\vec a}
|\inner{\partial_{\vec a}U}{\partial_{\vec a} V}|^2 
=& \sum_{\vec a} \abs{\Xi_{U^{*}V}(\vec a)}^{4}
\geq \abs{\Xi_{U^{*}V}(0)}^{4} \\
=& \abs{{\frac{1}{d^n} \Tr{U^*V}}}^4
=\abs{ \lra{U,V}}^4\;,
\end{align*}
as claimed in \eqref{LowerBoundQ2}.

The definition of $\norm{U}_{q^{k+1}}$, ensures that there exists a unitary $V\in C^{(k)}$, such that 
\begin{eqnarray}
\norm{U}_{q^{k+1}}
=|\inner{U}{V}|^2\;.
\end{eqnarray}
Also Theorem \ref{Thm:CliffordHierarchy} shows that $U\in C^{k}$ means that $\partial_{\vec a_k,...,\vec a_{1}}V$ is equal to $I$, up to a phase. 
Thus Proposition \ref{Prop:NormReductions} lets us write 
\begin{align*}
\norm{U}^{2^{k+1}}_{Q^{k+1}}
&= 
\mathbb{E}_{\vec a_1, ..., \vec a_{k}}
|\inner{\partial_{\vec a_k,...,\vec a_{1}}U}{\partial_{\vec a_k,...,\vec a_1} V}|^2\\
&\geq 
\mathbb{E}_{\vec a_1, ..., \vec a_{k-1}}
|\inner{\partial_{\vec a_{k-1},...,\vec a_{1}}U}{\partial_{\vec a_{k-1},...,\vec a_1} V}|^4\\
&\geq 
(\mathbb{E}_{\vec a_1, ..., \vec a_{k-1}}|\inner{\partial_{\vec a_{k-1},...,\vec a_{1}}U}{\partial_{\vec a_{k-1},...,\vec a_1} V}|^2)^2\\
&\geq ...
\geq \lrp{\abs{\inner{U}{V}}^2}^{2^k}.
\end{align*}
In the second inequality, we use the Schwarz inequality. Iterating this procedure establishes the desired bound.
\end{proof}

\section{q-HOFA Gives  a Clifford Hierarchy Test}
We consider an important task in quantum property testing, called Clifford hierarchy testing. 
Given a unitary $U$ and an integer $k$, the goal is to determine whether $U$ belongs to the $k$-Clifford hierarchy $\mathcal{C}^{(k)}$, or if it is $\epsilon$-far from $\mathcal{C}^{(k)}$. This task can be regarded as the quantum counterpart of low-degree testing of polynomials in classical theory.

 Clifford testing and magic entropy have been proposed by the authors using their quantum convolution~\cite{BGJ23c}. 
Let us first consider the $n$-qudit system with $d$ being odd prime.  The Hadamard convolution is given in  Definition \ref{Def:Had_Conv}.
We can generalize the test to any $k$-Clifford hierarchy,  using convolution-swap testing.

\begin{center}
  \begin{tcolorbox}[width=11cm,height=6 cm,title=$k^{\rm th}$  Clifford-Hierarchy Testing for odd prime $d$]
1. Chose $\vec a_1,...,\vec a_{k-2}$ from $V^n$ independently and uniformly and random. \\
  2. Prepare 4 copies of the Choi states $J_{\partial_{\vec a_1,...,\vec a_{k-2}}U}$ for the unitary  $\partial_{\vec a_1,...,\vec a_{k-2}}U$, and apply the convolution 
 for each 2 copies of $J_{\partial_{\vec a_1,...,\vec a_{k-2}}U}$, and get 2 copies of $J_{\partial_{\vec a_1,...,\vec a_{k-2}}U}\boxtimes_H J_{\partial_{\vec a_1,...,\vec a_{k-2}}U}$;\\
3.Perform the swap test for the 2 copies of $J_{\partial_{\vec a_1,...,\vec a_{k-2}}U}\boxtimes_H J_{\partial_{\vec a_1,...,\vec a_{k-2}}U}$.
If the output is $0$, it passes the test; otherwise, it fails.
    \end{tcolorbox}
   \end{center}

\begin{thm}\label{thm:QU_CH}
      The probability of acceptance for the  above  $k$th-Clifford hierarchy testing 
      can be written in terms of the quantum uniformity norm:
\begin{align}
P_{k}[U]=\frac{1}{2}
\left[1+\norm{U}^{2^k}_{Q^{k+1}}
\right].
\end{align}
\end{thm}
\begin{proof}
Based on the protocol for the $k$th-Clifford hierarchy testing, the probability can be written as

\begin{align*}
     &P_{k}[U]\\
     =&\frac{1}{2}\left[1+\underset{
     \vec a_1,...,\vec a_{k-2}}{\mathbb{E}}
\Tr{(
J_{\partial_{\vec a_1,...,\vec a_{k-2}}U}\boxtimes_HJ_{\partial_{\vec a_1,...,\vec a_{k-2}}U}
)^2
}\right].
\end{align*}

Then using Proposition \ref{prop:con_unif}, we obtain the result.
\end{proof}

 We can use other quantum convolutions, in addition to  the Hadamard convolution, to implement the Clifford hierarchy testing.
 In the $n$-qubit case, we need to change the definition of quantum convolution.
 
  \begin{Def} 
  The convolution of three $n$-qubit states $\rho_1, \rho_2, \rho_3$
is 
\begin{eqnarray}
\boxtimes_3(\rho_1, \rho_2, \rho_3)
=\Ptr{1^c}{V\rho_1\ot \rho_2\ot \rho_3 V^*},
\end{eqnarray}
where $V=U^{\ot n}=U^{(1,n+1, 2n+1)}\ot U^{(2, n+2, 2n+2)}\ot....\ot U^{(n,2n, 3n)}$, and 
$U$ is a $K$-qubit unitary constructed using CNOT gates:
\begin{eqnarray}
U:=\left(\prod^3_{j=1}CNOT_{j\to 1}\right)\left(\prod^3_{i=1}CNOT_{1\to i}\right),
\end{eqnarray}
and 
$
CNOT_{2\to 1}\ket{x}\ket{y}=\ket{x+y}\ket{y}
$ for any $x,y\in\mathbb{F}_2$. 
  \end{Def}
  
This gives rise to a Clifford hierarchy test for qubits.   The probability of acceptance for the above  $k$th-Clifford hierarchy testing is 
  \begin{align}
      P_{k}[U]=\frac{1}{2}\left[1+\mathbb{E}_{\vec a_1,...,\vec a_{k-2}\in V^n}
\Tr{(
\boxtimes_3J_{\partial_{\vec a_1,...,\vec a_{k-2}U}}
)^2
}\right].
  \end{align}
  
  \begin{cor}
  The success probability of the $k$-th level of Clifford hierarchy testing   is lower bounded by the maximal overlap with the $k$-th Clifford hierarchy,
  \begin{eqnarray}
      P_{k}(U)
\geq \frac{1}{2}\left[1+\norm{U}^{2^{k+1}}_{q^{k+1}}\right].
  \end{eqnarray}
\end{cor}

\section{Outlook}
In this work, we propose a framework of \textit{quantum} higher-order Fourier analysis and show its application in quantum computation. There are still many interesting open problems.  The inverse quantum uniformity norm conjecture  can be regarded as a quantum generalization of  the higher-order  Goldreich-Levin algorithm. This conjecture is:  if  $\norm{U}_{Q^3}\geq c$, then there exists a Clifford unitary $V$ such that the overlap between $U$ and $V$ is bounded below by some constant, independent of the number of qudits. If this conjecture is true, it is natural to ask: can one find an algorithm that makes polynomially-many queries of a given unitary $U$ and produces a decomposition of $U$ as a sum of Clifford unitaries and a small error term? 
In general, if $\norm{U}_{Q^{k+1}}\geq c$, 
does there exist a 
unitary $V$ in $k$th-level of Clifford hierarchy, such that the overlap between $U$ and $V$ is bounded below  by some constant, which 
is independent of the number of qudits?

Quantum teleportation is a crucial element that allows universal fault-tolerant quantum computation through stabilizer codes~\cite{GCh99}. An important problem is to find the depth of teleportation to implement quantum gates in the Clifford hierarchy, where the depth of teleportation
is a measure of the complexity of the gate to characterize the number of teleportation steps to implement a given gate in fault-tolerant quantum computation~\cite{ZengPRA08}. Is there a connection between the quantum uniformity measures and teleportation depth?

\section*{Acknowledgement}
We thank Chi-Ning Chou, Michael Freedman, Roy Garcia,  Cassandra Marcussen, Graeme Smith, Madhu Sudan for helpful discussion. We thank Arkopal Dutt for commenting on an  earlier draft and for bringing to our attention work by him and his collaborators using  classical uniformity norms. This
work was supported in part by JobsOhio GR138220,  ARO Grant W911NF-19-1-0302, ARO MURI Grant W911NF-20-1-0082, and NSF Eager
Grant 2037687.

\bibliographystyle{plain}
\bibliography{reference}
\bigskip\bigskip
%
 
\begin{appendix}
\begin{center}
APPENDIX
\end{center}

Here we derive some assertions in the main text that require extra space. The propositions are numbered here with reference to the corresponding statements in the main text.

\section{Generalized Quantum  Schwarz Inequality}
\setcounter{thm}{40}
\begin{prop}\label{AppendixProof41}
For $k=3$, 
\begin{eqnarray}
|\langle \set{B_{\vec u}}_{\vec u\in \set{0,1}^{k}}\rangle_{Q^{k}}|
\leq \Pi_{\vec u\in \set{0,1}^k}\norm{B_{\vec u}}_{Q^k}\;.
\end{eqnarray}
In particular, 
\begin{align*}
   &\abs{ \langle B_{000}, B_{001}, B_{010}, B_{011}, B_{100}, B_{101}, B_{110}, B_{111}\rangle_{Q^3}  }  \\
  & \quad\quad \leq\norm{B_{000}}_{Q^{3}} \, \norm{B_{001}}\, \norm{B_{010}}_{Q^{3}} \, \norm{B_{011}}_{Q^{3}} \, \norm{B_{100}}_{Q^{3}} \, \norm{B_{101}}_{Q^{3}} \, \norm{B_{110}}_{Q^{3}} \, \norm{B_{111}}_{Q^{3}}\;.
\end{align*}    
\end{prop}
\begin{proof}
Start by using the Schwarz inequality of Proposition 36 to give 
\begin{align}\label{FirstBoundFor8}
&\abs{ \langle B_{000}, B_{001}, B_{010}, B_{011}, B_{100}, B_{101}, B_{110}, B_{111}\rangle_{Q^3}  } \\
\nonumber\leq& \abs{\langle{B_{000}, B_{001}, B_{010}, B_{011}, B_{000}, B_{001}, B_{010}, B_{011}}\rangle_{Q^3}}^{1/2}\,\\
\nonumber&\times\abs{\langle{B_{100}, B_{101}, B_{110}, B_{111}, B_{100}, B_{101}, B_{110}, B_{111}}\rangle_{Q^3}}^{1/2}.
\end{align}
Now one can apply the second identity in Lemma 39, followed by the Schwarz inequality of Proposition 36 which  shows that the  first term in \eqref{FirstBoundFor8} satisfies 
\begin{align}\label{SecondBoundFor8}
    &\abs{\langle{B_{000}, B_{001}, B_{010}, B_{011}, B_{000}, B_{001}, B_{010}, B_{011}}\rangle_{Q^3}}^{1/2}\nonumber\\
 &\quad\quad = \abs{\langle{B_{011}, B_{010}, B_{011}, B_{010}, B_{001}, B_{000}, B_{001}, B_{000}}\rangle_{Q^3}}^{1/2} \nonumber\\
&\leq \abs{\langle{B_{011}, B_{010}, B_{011}, B_{010}, B_{011}, B_{010}, B_{011}, B_{010}}\rangle_{Q^3}}^{1/4}\nonumber\\
&\quad\quad \times
 \abs{\langle{ B_{001}, B_{000}, B_{001}, B_{000}, B_{001}, B_{000}, B_{001}, B_{000}}\rangle_{Q^3}}^{1/4}\;.
\end{align}
One can continue this process. For the first term on the right of \eqref{SecondBoundFor8},  
\begin{align}\label{FourthBoundFor8}
   &\abs{\langle{B_{011}, B_{010}, B_{011}, B_{010}, B_{011}, B_{010}, B_{011}, B_{010}}\rangle_{Q^3}}^{1/4}\nonumber\\
&\quad\quad =\abs{\langle{B^*_{011}, B^*_{011}, B^*_{010}, B^*_{010}, B^*_{011}, B^*_{011}, B^*_{010}, B^*_{010}}\rangle_{Q^3}}^{1/4}\nonumber\\
& \quad\quad =\abs{\langle{B^*_{011}, B^*_{011}, B^*_{011}, B^*_{011}, B^*_{010}, B^*_{010},B^*_{010}, B^*_{010}}\rangle_{Q^3}}^{1/4}\nonumber\\
& \quad\quad\leq \abs{\langle{B^*_{011}, B^*_{011}, B^*_{011}, B^*_{011}, B^*_{011}, B^*_{011}, B^*_{011}, B^*_{011}}\rangle_{Q^3}}^{1/8}\nonumber\\
&\quad\quad\quad\quad\times
\abs{\langle{B^*_{010}, B^*_{010}, B^*_{010}, B^*_{010}, B^*_{010}, B^*_{010}, B^*_{010}, B^*_{010}}\rangle_{Q^3}}^{1/8}\nonumber\\
&\quad\quad=\norm{B_{011}}_{Q^3}\norm{B_{010}}_{Q^3}\;.
\end{align}
Here the first equality comes from the first identity in Lemma 39, while the second equality comes from the second identity. One obtains the final inequality by using the quantum Schwarz inequality in Proposition 36.

The second term on the right of \eqref{SecondBoundFor8} satisfies a similar bound, so that 
\begin{align}\label{FifthBoundFor8}
    \abs{\langle{ B_{001}, B_{000}, B_{001}, B_{000}, B_{001}, B_{000}, B_{001}, B_{000}}\rangle_{Q^3}}^{1/4}
  \leq\norm{B_{000}}_{Q^3}\norm{B_{001}}_{Q^3}\;.
\end{align}
Inserting \eqref{FourthBoundFor8} and \eqref{FifthBoundFor8} into \eqref{SecondBoundFor8} gives 
\begin{align}\label{ThirdBoundFor8}
&\abs{\langle{B_{000}, B_{001}, B_{010}, B_{011}, B_{000}, B_{001}, B_{010}, B_{011}}\rangle_{Q^3}}^{1/2} \nonumber\\ 
&\quad\quad\leq
\norm{B_{000}}_{Q^3}\norm{B_{001}}_{Q^3}
\norm{B_{010}}_{Q^3}\norm{B_{011}}_{Q^3}\;.
\end{align}

The last term in \eqref{FirstBoundFor8} is similar to the term we just bounded, except it is a bound with $B_{1\vec{u}_2}$ replacing $B_{0\vec{u}_2}$. Making this substitution, one has 
\begin{align}\label{SeventhBoundFor8}
&\abs{\langle{B_{100}, B_{101}, B_{110}, B_{111}, B_{100}, B_{101}, B_{110}, B_{111}}\rangle_{Q^3}}^{1/2}\nonumber\\
&\quad\quad\leq
\norm{B_{100}}_{Q^3}\norm{B_{101}}_{Q^3}
\norm{B_{110}}_{Q^3}\norm{B_{111}}_{Q^3}\;.
\end{align}
Inserting the bounds \eqref{ThirdBoundFor8} and \eqref{SeventhBoundFor8} into \eqref{FirstBoundFor8} yields
 \begin{align*}\label{SixthBoundFor8}
&\abs{ \langle B_{000}, B_{001}, B_{010}, B_{011}, B_{100}, B_{101}, B_{110}, B_{111}\rangle_{Q^3}  }\\
&\quad\quad\leq
     \langle{\set{B_{\vec u}}_{\vec u\in \set{0,1}^3}}\rangle_{Q^3}
     \leq \Pi_{\vec u\in \set{0,1}^3}\norm{B_{\vec u}}_{Q^3}\;,
 \end{align*}
and completes the proof.
\end{proof}

\setcounter{thm}{44}

\begin{prop}\label{AppendixProof45}
Given two operators $A,B\in L(\mathcal{H}^{\ot n})$ and an integer $k\ge 1$,
one has for Schatten norms $\norm{\ \cdot\ }_r$ that
\begin{align}
\left( \mathbb{E}_{\vec a\in V^n} \| A(\vec a) B\|_r^s\right)^{1/s} \le \|A\|_p \|B\|_q,
\end{align}
where
\[ p = q =  \frac {2^k}{k + 1}, \quad
r = \frac {2^{k-1}}k, \quad
s = 2^{k-1}.\]
\end{prop}
\begin{proof}
We use  H\"{o}lder inequality for  $A_1,A_2,A_3$, namely 
\begin{align}\label{eq:holder_op}
 \|A_1A_2A_3\|_r \le \|A_1 \|_{p_1} \|A_2\|_{p_2}\|A_3\|_{p_3}\;,
\end{align}
where $ \frac1r = \frac1{p_1}+\frac1{p_2}+\frac1{p_3} $. 
We take 
$p_1=\frac{sp}{s-p}$, $p_2=s$, and $p_3=\frac{sq}{s-q}$.
We use the identity 
\[
\abs{A(\vec a)B}^{r}
=\lrp{\abs{A(\vec a)B}^{2}}^{r/2}
= \lrp{B^*\lrp{A^*A}(\vec a)B}^{r/2}\;.
\]
Suppose $B=\abs{B}V$ for unitary $V$. Then, for any $0<r$, we have
$$\Tr{\abs{A(\vec a)B}^{r}} = \Tr{\lrp{\abs{B}\abs{A}^{2}(\vec a)\abs{B}}^{r/2}}=\Tr{\lrp{\abs{A}(\vec a)\abs{B}}^{r}}.$$
Hence,
\begin{align*}
&\left(\frac1{d^n} \Tr{ \abs{A(\vec a) B}^r }\right)^{\frac{1}{r}}\\
=& \left(\frac1{d^n} \Tr{ \lrp{\abs{A}(\vec a) \abs{B}}^r }\right)^{\frac{1}{r}}
= \left(\frac1{d^n} \Tr{\lrp{\abs{A}^{\frac{s-p}{s}}(\vec a) \abs{A}^{\frac{p}{s}}(\vec a)\abs{B}^{\frac{q}{s}}   \abs{B}^{\frac{s-q}{s}}}^r }\right)^{\frac{1}{r}}\\
\le&  \left(\frac1{d^n} \Tr{(|A|^{\frac{s-p}{s}})^{\frac{sp}{s-p}} }\right)^{\frac{s-p}{sp} } 
\left(\frac1{d^n} \Tr{\big||A|^{\frac{p}{s}} (\vec a) |B|^{\frac{q}{s}}\big|^{s} }\right)^{\frac{1}{s} }
\left(\frac1{d^n} \Tr{(|B|^{\frac{s-q}{s}})^{\frac{sq}{s-q}} }\right)^{\frac{s-q}{sq} }\;.
\end{align*}
In the last inequality we use \eqref{eq:holder_op}.
By rewriting it, we have
\begin{align}\label{230810eq2}
\| A (\vec a) B\|_r^s \le \|A\|_p^{s-p} \|B\|_q^{s-q} \frac1{d^n} \Tr{\big|w(\vec a) |A|^{\frac{p}{s}} w(\vec a)^* |B|^{\frac{q}{s}}\big|^{s} }\;.
\end{align}

If $k=1$, 
\begin{align*}
\mathbb{E}_{\vec a\in V^n}\frac1{d^n} \Tr{\big|w(\vec a) |A|^{\frac{p}{s}} w(\vec a)^* |B|^{\frac{q}{s}}\big|^{s} } 
= \mathbb{E}_{\vec a\in V^n} \frac1{d^n} \Tr{w(\vec a)|A|w(\vec a)^* |B |  } = \|A\|_p^p \|B\|_q^q\;.
\end{align*}
When $k\ge2$, 
\begin{align*}
&\mathbb{E}_{\vec a\in V^n}\frac1{d^n} \Tr{\big|w(\vec a) |A|^{\frac{p}{s}} w(\vec a)^* |B|^{\frac{q}{s}}\big|^{s} }\\
=& \mathbb{E}_{\vec a\in V^n}\frac1{d^n} \Tr{ \left(w(\vec x)  |A| ^{\frac{p}{s}}  w(\vec a)^*  |B|^{\frac{q}{s}} |B|^{\frac{q}{s}} w(\vec a)  |A|^{\frac{p}{s}}  w(\vec a)^* \right)^{\frac{s}2}} \\ 
=& \mathbb{E}_{\vec a\in V^n}\frac1{d^n} \Tr{ \left(w(\vec a)  |A| ^{\frac{2p}{s}}  w(\vec a)^*  |B|^{\frac{2q}{s}} \right)^{\frac{s}2}} \\
\le & \||A|^{\frac{2p}{s}}\|_{s/2}^{s/2} \; \||B|^{\frac{2q}{s}}\|_{s/2}^{s/2}\\
= & \|A\|_p^p \|B\|_q^q\;.
\end{align*}
Here the inequality comes from the Lemma 43.
Combined with \eqref{230810eq2}, we get
\begin{align*}
  \mathbb{E}_{\vec a\in V^n} \|  A(\vec a) B\|_r^s  \le \|A\|_p^s \|B\|_q^s\;,
\end{align*}
which completes the proof.
\end{proof}

\section{Calculation of the Quantum Measures for Gates in \S 7 }
\label{appen:example_de}
We give the detailed calculations of the quantum uniformity measures for several examples of quantum gates, most of which are stated in \S 7. 

\begin{Exam}[\bf The Weyl Gate]
For ${b}=({p},  {q})$, the one-qudit Weyl (or Pauli) gates are $w(a)=\zeta^{-pq}Z^pX^q$.  The $n$-qudit Weyl gates are tensor products of 1-qudit gates, so their quantum uniformity  measures are the $n^{\rm, th}$ power of the one-qudit measures.  We claim that 
%
\begin{align}
  \label{eq:UN1_pauli}      \norm{w(\vec b)}_{Q^1}=&\delta_{\vec b, \vec 0}\;,
\quad\text{and}\quad
\norm{w(\vec b)}_{Q^k}=1, \quad \forall k\geq 2. 
\end{align}

To calculate the  $k=1$ norm,  use  Proposition 25. In the computational basis, $Z$ is diagonal and $X$ is off-diagonal. So  $\Tr{Z^{p}X^{q}}=\Tr{Z^{p}}\delta_{q,0}=d\delta_{q,0}$. Thus, 
\begin{align*}
    \norm{w(\vec b)}_{Q^1}
    =\abs{\frac{1}{d^n}\Tr{w(\vec b)}}
    =\delta_{\vec b, \vec 0}\;.
\end{align*}

To calculate the $k=2$ norm, observe that  
\begin{align*}
       \norm{w(\vec b)}^4_{Q^2}
       =&\mathbb{E}_{\vec a}\norm{\partial_{\vec a}w(\vec b)}^2_{Q_1}
       =\mathbb{E}_{\vec a}
       \abs{\frac{1}{d^n}\Tr{\partial_{\vec a}w(\vec b)}}^2\\
       =&\mathbb{E}_{\vec a}\abs{ \frac{1}{d^{ n}}\Tr{w(\vec a)w(\vec b)w(\vec a)^* w(\vec b)^*}}^2
       =\mathbb{E}_{\vec a} \abs{\frac{1}{d^n}\Tr{I}}^2
      =1.
\end{align*}

Finally, $\norm{w(\vec b)}_{Q^k}=1$ for $k\geq 3$, as 
$\norm{w(\vec b)}^{2^k}_{Q^k}
=\mathbb{E}_{\vec a}\norm{\partial_{\vec a}w(\vec b)}^{2^{k-1}}_{Q_{k-1}}
=\mathbb{E}_{\vec a}\norm{I}^{2^{k-1}}_{Q^{k-1}}=1
$.
\end{Exam}

\begin{Exam}[\bf The Fourier Gate]
The Fourier gate on an $n$-qudit system plays an important role in quantum computation, 
\[
F=\frac{1}{{d^{n/2}}}\sum_{\vec {x}, \vec{y}}\omega^{\vec {x} 
 \cdot\vec{y}}\ket{\vec x}\bra{\vec y}\;,
 \quad{where}\quad
 \omega=e^{2\pi i/d}\;.
\] 
The quantum uniformity measures of $F$ are
\begin{align}\label{QMeasuresForWeyl}
        \norm{F}_{Q^1}
          &=\left\{
    \begin{aligned}
            &\,0\;, \text{ if}~ d=2\\
      &\frac{1}{d^n},  \text{ if $d$ is odd prime}
    \end{aligned}
    \right.
    \;,
    \\
 \norm{F}_{Q^2} &=\left\{
    \begin{aligned}
        \frac{1}{2^{n/4}}, &\quad \text{if}~~ d=2 \\
        \frac{1}{d^{n/2}}, &\quad \text{if d is odd prime}\\
    \end{aligned}
    \right.\;,
    \\
\norm{F}_{Q^k}&=1, \quad \forall k\geq 3\;.
\end{align}

In verifying \eqref{QMeasuresForWeyl} we   use that the $n$-qudit $F$ is the tensor product of 1-qudit transforms $F_1$. 
Thus $\norm{F}_{Q^k}=\norm{F_1}_{Q^k}^{n}$. 

\paragraph{\bf The case $k=1$:}
By Proposition 25, 
\begin{align}\label{Q1-Measure}
    \norm{F_{1}}_{Q^1}
=\frac{1}{d}\abs{\Tr{F_{1}}}.
\end{align}
The trace $\Tr{F_{1}}$ can be evaluated in the computational basis as a Gauss sum:
\begin{align}
  \Tr{F_{1}}= {\frac{1}{\sqrt{d}}\sum_{x\in\mathbb{Z}_d}\omega^{x^2}}
  &=\left\{
    \begin{aligned}
        1+i, &\quad \text{if}~~ d\equiv 0\mod 4 \\
        1, &\quad \text{if}~~ d\equiv 1\mod 4\\
        0, &\quad \text{if}~~ d\equiv 2\mod 4\\
        i, &\quad \text{if}~~d\equiv 3\mod 4
\end{aligned}
\right.\;,
\quad\text{so}\\
\abs{\Tr{F_{1}}}
  &=\left\{
    \begin{aligned}
        0\;, &\quad \text{if}~~ d=2\\
        1\;, &\quad \text{if $d$ is an odd prime}
\end{aligned}
\right.\;.
\end{align}
Inserting this into \eqref{Q1-Measure}  gives for the one-qudit system, 
\begin{align*}\label{FourierQ1}
        \norm{F_1}_{Q^1}
          &=\left\{
    \begin{aligned}
            \,0\;,\phantom{,} &\quad \text{if}~~ d=2\\
      \frac{1}{d}\;, &\quad \text{if $d$ is an odd prime}\\
    \end{aligned}
       \right.\;,
\end{align*}
{and for $n$ qudits},
\begin{align*}
               \norm{F}_{Q^1}
          &=\left\{
    \begin{aligned}
            \,0\;,\phantom{,} &\quad \text{if}~~ d=2\\
      \frac{1}{d^{n}}\;, &\quad \text{if $d$ is an odd prime}\\
    \end{aligned}
       \right.\;.
    \end{align*}

\paragraph{\bf The case $k=2$:}
Let $a=(p,q)\in \mathbb{Z}_d\times\mathbb{Z}_d$. We claim  that
$\partial_{\vec a}F_{1}$ is equal to 
 $w(p-q,p+q)$, up to some phase.  In fact, 
\begin{align*}
\partial_{a}F_{1}
=&w(a)F_{1}w(a)^* F_{1}^*
=Z^pX^q F_{1}X^{-q}Z^{-p}F_{1}^*
=Z^pX^q (F_{1}X^{-q}F_{1}^*)(F_{1}Z^{-p}F_{1}^*)\\
=&\omega^{q^{2}} Z^{p-q}X^{p+q}
=\zeta^{p^2+q^2}w(p-q,p+q)\;.
\end{align*}
Here we used the fact that 
$F_{1}ZF_{1}^* =X^{-1}$, and $F_{1}XF_{1}^*=Z$. Then  $\norm{F_{1}}_{Q^2}$ can be computed using the recursion relation in the main text,  namely Proposition 2.  This gives,
\begin{align*}
    \norm{F_{1}}^4_{Q^2}=&
    \mathbb{E}_{ a} \norm{\partial_{ a}F_{1}}^2_{Q^1}
    =\frac{1}{d^2}\sum_{p,q\in \mathbb{Z}_d}\norm{w(p-q,p+q)}^2_{Q^1}
    =\frac{1}{d^2}\sum_{p,q\in \mathbb{Z}_d}\delta_{p+q,0}\,\delta_{p-q,0}\;,
\end{align*}
where the last equality comes from  \eqref{eq:UN1_pauli}. This is
also equal to 
\begin{equation}
 \norm{F_{1}}^4_{Q^2} =\left\{
    \begin{aligned}
        \frac{1}{2}, &\quad \text{if}~~ d=2 \\
        \frac{1}{d^2}, &\quad \text{if d is an odd prime}
    \end{aligned}
    \right.\;.
    \end{equation}
    \quad\text{Hence}\quad
\begin{equation}
     \norm{F}_{Q^2} =\left\{
    \begin{aligned}
        \frac{1}{2^{n/4}}, &\quad \text{if}~~ d=2 \\
        \frac{1}{d^{n/2}}, &\quad \text{if d is an odd prime}\\
    \end{aligned}
    \right.\;.
\end{equation}

\paragraph{\bf The case $k\geq 3$:}
Again one can use Proposition 2 to show  
\begin{align*}
     \norm{F_{1}}^8_{Q^3}=   \mathbb{E}_{ a}\norm{\partial_{a} F_{1}}^4_{Q^2}
     =\frac{1}{d^2}\sum_{p,q\in\mathbb{Z}_d}\norm{w(p-q,p+q)}^4_{Q^2}
     =1\;.
\end{align*}
Hence  $\norm{F}_{Q^3}=\norm{F_1}_{Q^3}=1$.
The same argument  shows that $\norm{F}_{Q^k}=1$ for any $k > 3$.
\bigskip

For $d=2$ case, namely for a qubit system, the Fourier gate $F$ is known as  Hadamard gate $H$. For $n=1$,  
\begin{equation}
H_1=\frac{1}{\sqrt{2}}
\left[
\begin{array}{cccc}
1&1\\
1&-1\\
\end{array}
\right].
\end{equation}
The corresponding quantum uniformity measures are 
\begin{align}
\norm{H_{1}}_{Q^1}&=0\;, \quad \norm{H_{1}}_{Q^2}=\frac{1}{2^{1/4}}\;,
\quad
\norm{H_{1}}_{Q^k}=1\;, \quad \forall k\geq 3\;.
\end{align}


\end{Exam}

\begin{Exam}[\bf The CNOT gate]
Let us consider a two-qudit ($n=2$) CNOT gate $$CNOT=\sum_{x,y\in \mathbb{Z}_d}\proj{x}\ot \ket{x+y}\bra{y}\;.$$  This is an important gate for the generation of quantum entanglement.  Here we compute  the corresponding quantum uniformity norms 
\begin{align}
\norm{CNOT}_{Q^1}&=\frac{1}{d}\;, \quad \norm{CNOT}_{Q^2}=\frac{1}{\sqrt{d}}\;,
\quad
\norm{CNOT}_{Q^k}=1\;, \quad \forall k\geq 3.
\end{align}

\paragraph{\bf The case $k=1$:}
Using  Proposition 2 or Proposition 25 in the main text for the case $n=2$, and the relation $\Tr{CNOT}=d$, we have  
\begin{align*}
\norm{CNOT}_{Q^1}
    =\frac{1}{d^2}|\Tr{CNOT}|
    =\frac{1}{d}\;. 
\end{align*}

\paragraph{\bf The case $k=2$:}
It is convenient to know the action of CNOT on the Weyl operator $w(\vec a)=w(a_{1})\otimes w(a_{2})$, where $\vec a=(a_{1}, a_{2})$ and $a_{j}=(p_{j},q_{j})$.  We claim that 
\begin{align}\label{CNOT-on-w}
CNOT \, w(\vec a)  \, CNOT^{*}
&= CNOT\lrp{w(p_{1},q_{1})\ot w(p_{2},q_{2})}CNOT^{*}\nonumber\\
&= \zeta^{p_{2}(q_{1}+q_{2})} w(p_{1}-p_{2},q_{1}) \ot w(p_{2}, q_{1}+q_{2})\;.
\end{align}
As 
\be
w(p_{1},q_{1})\ot w(p_{2},q_{2}
= w(\vec a)
= \zeta^{-p_{1}q_{1}+p_{2}q_{2}} \lrp{  Z^{p_{1}}\otimes X^{q_{2}} }   \lrp{   X^{q_{1}} \otimes Z^{p_{2}}    }\;,
\ee
one can use  the actions
\begin{align}\label{CNOTAction-1}
      CNOT\, (Z^{p}\ot X^{q})\, CNOT^*
      =Z^{p}\ot X^{q}\;,
      \end{align}
      and
      \begin{align}
    CNOT\, (X^{q}\ot Z^{p}) \,CNOT^*
    =X^{q}Z^{-p}\ot X^{q}Z^{p}
\end{align}
to obtain \eqref{CNOT-on-w}.    The first identity in \eqref{CNOTAction-1} comes from  
\begin{align*}
&CNOT
\, (Z^{p}\ot X^{q})\, CNOT^*\nonumber\\
& \quad\quad=  
\sum_{a,b,c,d\in \mathbb{Z}_d}
 \lrp{\proj{a} \ot \ket{a+b}\bra{b}}
\, (Z^{p} \ot X^{q})\, 
\lrp{  \proj{c} \ot  \ket{d} \bra{c+d}}\\
& \quad\quad=  
 \sum_{a,b,c,d\in \mathbb{Z}_d} \omega^{pc} 
  \lrp{\proj{a} \ot \ket{a+b}\bra{b}}
\lrp{  \proj{c} \ot  \ket{d+q} \bra{c+d}}\\
& \quad\quad=  
 \sum_{a,b,c,d\in \mathbb{Z}_d} \omega^{pc} 
 \delta_{a,c}  \delta_{b,d+q}  \ket{a} \bra{c} \ot  \ket{a+b}  \bra{c+d}\\
& \quad\quad=  
 \sum_{a,b\in \mathbb{Z}_d} \omega^{pa} 
   \ket{a} \bra{a} \ot  \ket{a+b}  \bra{a+b-q}\\
& \quad\quad=  
 \sum_{a,b\in \mathbb{Z}_d} \omega^{pa} 
   \ket{a} \bra{a} \ot  \ket{a+b+q}  \bra{a+b}\\
& \quad\quad=  
 \sum_{a,b\in \mathbb{Z}_d} \omega^{pa} 
   \ket{a} \bra{a} \ot  \ket{b+q}  \bra{b}
= Z^{p} \ot X^{q}\;.
\end{align*}
The second identity  in \eqref{CNOTAction-1} is
\begin{align*}
&CNOT
\, ( X^{q}   \ot  Z^{p})\, CNOT^*\\
& \quad\quad=   
\sum_{a,b,c,d\in \mathbb{Z}_d}
 \lrp{\proj{a} \ot \ket{a+b}\bra{b}}
\, ( X^{q}   \ot  Z^{p})\, 
\lrp{  \proj{c} \ot  \ket{d} \bra{c+d}}\\
& \quad\quad=  
 \sum_{a,b,c,d\in \mathbb{Z}_d} \omega^{pd} 
  \lrp{\proj{a} \ot \ket{a+b}\bra{b}}\,
\lrp{  \ket{c+q} \bra{c} \ot  \ket{d} \bra{c+d}}\\
& \quad\quad=  
 \sum_{a,b,c,d\in \mathbb{Z}_d} \omega^{pd} 
 \delta_{a,c+q}  \delta_{b,d}  \ket{a} \bra{c} \ot  \ket{a+b}  \bra{c+d}\\
& \quad\quad=   
 \sum_{a,b\in \mathbb{Z}_d} \omega^{pb} 
   \ket{a} \bra{a-q} \ot  \ket{a+b}  \bra{a+b-q}\\
& \quad\quad=  
 \sum_{a,b\in \mathbb{Z}_d} \omega^{pb} 
   \ket{a+q} \bra{a} \ot  \ket{a+b+q}  \bra{a+b}\\
& \quad\quad=  
 \sum_{a,b\in \mathbb{Z}_d} \omega^{p(b-a)} 
   \ket{a+q} \bra{a} \ot  \ket{b+q}  \bra{b}
= X^{q} Z^{-p} \ot X^{q}  Z^{p}\;.
\end{align*}
Hence
\begin{align*}
CNOT\, w(\vec a)\, CNOT^{*}
&= \zeta^{-p_{1}q_{1}+ p_{2}q_{2}  }   (Z^{p_{1}} \ot X^{q_{2}}) \, (X^{q_{1}}Z^{-p_{2}}\ot X^{q_{1}} Z^{p_{2}})\\
&= \zeta^{-p_{1}q_{1}+ p_{2}q_{2}  }  \zeta^{2q_{1}p_{2}   - (q_{1}+q_{2})p_{2}}  Z^{p_{1}-p_{2}}  X^{q_{1}} \ot  Z^{p_{2}}X^{q_{1}+q_{2}}\\
&= \zeta^{p_{2}(q_{1}+q_{2})} w(p_{1}-p_{2},q_{1}) \ot w(p_{2}, q_{1}+q_{2})\;.
\end{align*}
Therefore
\be
\partial_{\vec a} \,CNOT
= \zeta^{p_{2}(q_{2}-q_{1})} Z^{p_{2}}  \ot X^{-q_{1}}  \;,
\ee
which follows from 
\begin{align}
\partial_{\vec a} \,CNOT
&=w(\vec a)\, CNOT \, w(\vec a)^{*}\,CNOT^{*}
\nonumber\\
&=  \zeta^{p_{2}(q_{1}+q_{2})} w(p_{1},q_{1})w(p_{2}-p_{1} , -q_{1})  \ot w(p_{2},q_{2}) w(-p_{2},-q_{1}-q_{2})\nonumber\\
&=\zeta^{   p_{2}(q_{1}+q_{2})   +(- p_{1} q_{1} -  q_{1}(p_{2}-p_{1})  )  -p_{2}   (q_{1}+q_{2})      +q_{2}p_{2}    }\,  w(p_{2},0) \ot w(0,-q_{1})\nonumber\\
& = \zeta^{p_{2}(q_{2}-q_{1})}  w(p_{2},0) \ot w(0,-q_{1})
= \zeta^{p_{2}(q_{2}-q_{1})} Z^{p_{2}}  \ot X^{-q_{1}}  \;.
\label{partialCNOT}
\end{align}
Thus 
\be
\norm{\partial_{\vec a}\,CNOT}_{Q^{1}}
=\abs{\frac{1}{d^{2}}\Tr{\partial_{\vec a}\,CNOT}}
=   \delta_{p_{2},0}  \delta_{q_{1},0}\;.
\ee
It then follows from Proposition 2 in the main text that 
\be
\norm{CNOT}_{Q^{2}}^{4}
=  \mathbb{E}_{\vec a} \,{\norm{\partial_{\vec a}\,CNOT }}_{Q_{1}}^{2}
= \frac{1}{d^{2}}\;,
\quad\text{and}\quad
\norm{CNOT}_{Q^{2}} = \frac{1}{d^{1/2}}\;.
\ee

\paragraph{\bf The case $k\geq 3$:}
For $k=3$, we use the fact that in \eqref{partialCNOT} we show that $\partial_{\vec a} \,CNOT$ is a Weyl operator, up to a phase. Thus $\norm{\partial_{\vec a}\,CNOT}_{Q^{2}}=1$. Thus  using Proposition 2,
\begin{align*}
\norm{CNOT}^8_{Q^3}=\mathbb{E}_{\vec a} \norm{\partial _{\vec a} \,CNOT   }_{Q^{2}}^{2}
=1\;.
\end{align*}
Then  $\norm{CNOT}_{Q^{k}}=1$ also for  $k\geq 4$.

\paragraph{The qubit case:}
For $d=2$ case, i.e., 1-qubit system, the CNOT gate can be written in the computational basis as the following
matrix, 
\begin{equation}
CNOT=
\left[
\begin{array}{cccc}
1&0&0&0\\
0&1&0&0\\
0&0&0&1\\
0&0&1&0\\
\end{array}
\right].
\end{equation}
The corresponding quantum uniformity norms are
\begin{align}
\norm{CNOT}_{Q^1}=\frac{1}{2}, \quad \norm{CNOT}_{Q^2}=\frac{1}{\sqrt{2}},\quad
\norm{CNOT}_{Q^k}=1, \quad \forall k\geq 3.
\end{align}
\end{Exam}

\begin{Exam}[\bf The one-qubit T gate]
A single-qubit gate $T$ gate in $\mathcal{C}^{3}$  plays important role in the universal quantum computation, as
Clifford unitary + T gate can generate any unitary operator.  The $T$ gate is a fourth root of $Z$ and  has the matrix form
\begin{equation}
T=\left[
\begin{array}{cc}
1 &0\\
0& e^{i\pi/4}
\end{array}
\right].
\end{equation}
The quantum uniformity norms of $T$ are 
\begin{align}   
\norm{T}_{Q^1}
=\frac{\sqrt{2+\sqrt{2}}}{2}, \quad
\norm{T}_{Q^2}
=\left(\frac{3}{4}\right)^{1/4},\quad
\norm{T}_{Q^3}
=\left(\frac{3}{4}
\right)^{1/8}, \quad
\norm{T}_{Q^k}=1, \forall k\geq 4.
\end{align}

\paragraph{\bf The case $k=1$:}
\begin{align*}
   \norm{T}_{Q^1} 
   =\frac{1}{2}|\Tr{T}|
   =\frac{|1+e^{i\pi/4}|}{2}=\frac{\sqrt{2+\sqrt{2}}}{2}
   = .92387953....
\end{align*}

\paragraph{\bf The case $k=2$:}

Since $T$ commutes with $Z$,  
\[
Z^{p}X^{q}TX^{-q}Z^{-p}
=\left\{  
\begin{matrix}         
	 \begin{pmatrix}
         1&{0}\\
         {0}&{e^{i\pi/4}}
         \end{pmatrix}
\phantom{xxxx}\text{ for }q=0\\
         \begin{pmatrix}
         {e^{i\pi/4}}&{0}\\
         {0}&{1}
         \end{pmatrix}
\phantom{xxxx}\text{ for $q=1$}
\end{matrix} 
\right.\;,
\]
and
\[
\partial_{p,q}\,T
= w(p,q)Tw(p,q)^{*}T^{*}
=\left\{  
\begin{matrix}         
	 \begin{pmatrix}
         1&{0}\\
         {0}&{1}
         \end{pmatrix}
\phantom{xxxxxxx}\text{ for }q=0\\
         \begin{pmatrix}
         {e^{i\pi/4}}&{0}\\
         {0}&{e^{-i\pi/4}}
         \end{pmatrix}
\text{ for $q=1$}
\end{matrix} 
\right.
\;. 
\]

Thus, we have
\be
\abs{\frac{1}{2}\Tr{\partial_{a} \,T}}
= \left\{  
\begin{matrix}         
	1
\phantom{xxxx}\text{ for }q=0\\
         \frac{1}{\sqrt{2}} 
         \phantom{xx,}\text{ for }q=1
\end{matrix} 
\right.\;. 
\ee
As a consequence, we can use Proposition 2 to show that 
\begin{align*}
\norm{T}^4_{Q^2}
&=
\mathbb{E}_{a}
\norm{\partial_{a}\,T}^2_{Q^1}
=\mathbb{E}_{a} \abs{\frac{1}{2} \Tr{\partial_{a}\,T}}^{2}
= \frac{1}{2} \lrp{ 1 + \frac{1}{2}     }
= \frac{3}{4}\;,
\end{align*}
Hence, we have
\begin{align*}
\norm{T}_{Q^{2}}  = \lrp{\frac{3}{4}}^{1/4} =.93060486....
\end{align*}

\paragraph{\bf The case $k=3$:}
By Proposition 2, we have
\begin{align}\label{TQ3}
\norm{T}^8_{Q^3}
&=
\mathbb{E}_{b}
\norm{\partial_{b}\,T}^2_{Q^2}
= \mathbb{E}_{b,a}
\norm{\partial_{a}\partial_{b}\,T}^2_{Q^1}\;.
\end{align}
Note that 
\begin{align}
\partial_{a}\partial_{b}\,T
&= \left\{
\begin{matrix}
\begin{pmatrix}
1&0\\
0&1
\end{pmatrix}
\;, \text{ if } (a,b) =(0,0), (0,1), \text{ or } (1,0)\\
\begin{pmatrix}
-i&0\\
0&i
\end{pmatrix}
 \;, \text{ if } (a,b)=(1,1) \phantom{xxxxxxxxxxxxx}
\end{matrix}
\right. \;.
\end{align}
Inerting this into \eqref{TQ3} shows that 
\be
\norm{T}_{Q^{3}}^{8} = \frac{3}{4}\;,
\quad\text{and}\quad
\norm{T}_{Q^{3}} =\lrp{ \frac{3}{4}}^{1/8} 
= .96467863....
\ee

\paragraph{\bf The case $k\geq 4$:}
In this case $\partial_{a}\partial_{b}\partial_{c}\,T=\pm I$ for all cases of $(a,b,c)$.  Thus 
\be
\norm{T}_{Q^{4}} = 1\;,
\ee
and the same holds for $\norm{T}_{Q^{k}}$ for $k>4$.
\end{Exam}

\begin{Exam}[\bf The CCZ (or Toffoli) gate]
The three-qudit gate $CCZ\in\mathcal{C}^{3}$
(or control-control $Z$ gate)
is defined as 
\begin{align}
    CCZ=\sum_{x,y,z=1}^{d-1}\omega^{xyz} \proj{x}\ot \proj{y}\ot\proj{z}\;,
    \quad\text{where}\quad
    \omega=e^{2\pi i/d}\;.
\end{align}
The $CCZ$ gate  plays an important role in the realization of universal quantum computation. The corresponding quantum uniformity norms are 
\begin{align}
&\norm{CCZ}_{Q^1}
=\frac{2d-1}{d^2}, \quad
\norm{CCZ}_{Q^2}
=\left(\frac{d^3+(d-1)^3+d[d^3-(d-1)^3-1]}{d^6}\right)^{1/4},\\
&\norm{CCZ}_{Q^3}=\left(\frac{d^3+d^2-1}{d^5}\right)^{1/8}, \quad\text{and}\quad
\norm{CCZ}_{Q^k}=1, \ \ \forall k\geq 4\;.
\end{align}

\paragraph{\bf The case $k=1$:}
One can establish this by using Proposition 2, 
\begin{align*}
\norm{CCZ}_{Q^1}
=&\frac{1}{d^3}|\Tr{CCZ}|
=\frac{1}{d^3}
\left|\sum_{x,y,z\in\mathbb{Z}_d}\omega^{xyz}_d\right|
=\frac{1}{d^2}\left|\sum_{x,y\in \mathbb{Z}_d}\delta_{xy,0}\right|
    =\frac{(2d-1)}{d^2}\;.
\end{align*}

\paragraph{\bf The case $k=2$:}
Since $CCZ$ is diagonal in the computational basis, it commutes with $Z_j$, for $j=1,2,3$.  Thus  using Proposition 2, and with 
$\vec a= ((p_1,q_1),(p_2,q_2),(p_3,q_3))\in \mathbb{Z}_{d}^{6}$,
\begin{align*}
\norm{CCZ}^4_{Q^2}
&=\mathbb{E}_{\vec a} \norm{\partial_{\vec a}  CCZ  } ^2_{Q_1}
=\mathbb{E}_{\vec a}\norm{w(\vec a) \, CCZ \, w(\vec a)^*\,CCZ^*}^2_{Q_1}\\
&=\mathbb{E}_{\vec a}\frac{1}{d^6} \left|\Tr{w(\vec a)\, CCZ \, w(\vec a)^*  \, CCZ^*}\right|^2\\
&=\mathbb{E}_{\vec a}\frac{1}{d^6}\left|\Tr{(X^{q_1}\ot X^{q_2}\ot X^{q_3})\, CCZ\,
(X^{-q_1}\ot X^{-q_2}\ot X^{-q_3})\, CCZ^*}\right|^2\;.
\end{align*}
Evaluating the trace in the computational basis $\ket{k_1}\ot\ket{k_2}\ot\ket{k_3}$, we obtain

\begin{align*}
\norm{CCZ}^4_{Q^2}
    &=\mathbb{E}_{ a, b,  c\in \mathbb{Z}_d}\frac{1}{d^6}
    \left|\sum_{x,y,z}\omega^{xyz-(x+a)(y+b)(z+c)}_d\right|^2\\
    &=\frac{d^3+(d-1)^3+d[d^3-(d-1)^3-1]}{d^6}.
\end{align*}
And 
\begin{align}
   \norm{CCZ}^8_{Q^3}  
   =\frac{d^3+3(d-1)d+3(d-1)^2d+(d-1)^3}{d^6}
\end{align}

Let us consider now the 3-qubit gate  $CCZ$, which has the following form  in the computational basis
\begin{equation}
CCZ=
\left[
\begin{array}{cccccccccccccccc}
1&0&0&0&0&0&0&0\\
0&1&0&0&0&0&0&0\\
0&0&1&0&0&0&0&0\\
0&0&0&1&0&0&0&0\\
0&0&0&0&1&0&0&0\\
0&0&0&0&0&1&0&0\\
0&0&0&0&0&0&1&0\\
0&0&0&0&0&0&0&-1\\
\end{array}
\right].
\end{equation}
Then we calculate the quantum uniformity norm, and find that 
\begin{align}
&\norm{CCZ}_{Q^1}
=\frac{3}{4}, \quad
\norm{CCZ}_{Q^2}
=\left(\frac{11}{32}\right)^{1/4},\\
&\norm{CCZ}_{Q^3}=\left(\frac{11}{32}\right)^{1/8}, \quad
\norm{CCZ}_{Q^k}=1,\quad \forall k\geq 4
\end{align}
\end{Exam}

\end{appendix}

\end{document}